\documentclass[11pt]{article}
\usepackage{amssymb,amsmath,amsfonts}
\usepackage{graphicx,color,enumitem}
\usepackage{amsthm} 
\usepackage{bm}
\usepackage[round]{natbib}
\usepackage{soul}
\usepackage{geometry}
\usepackage{stmaryrd}
\usepackage{pifont}     
\usepackage{bbold}
 \usepackage{mathtools}
 \mathtoolsset{showonlyrefs}
\usepackage[colorinlistoftodos, textwidth=4cm, shadow]{todonotes}
\RequirePackage[colorlinks,citecolor=blue,urlcolor=blue]{hyperref}
\usepackage{bbm}

\usepackage{soul}
\usepackage{ulem}

\usepackage{authblk}

\usepackage{subcaption}

\usepackage{empheq}
\definecolor{myblue}{rgb}{0.69, 0.806, 1}%{114, 133, 165}

\newcommand{\E}{{\mathbb E}}
\newcommand{\F}{{\mathbb F}}

\newcommand{\R}{{\mathbb R}}

\newcommand{\Acal}{{\mathcal A}}

\newcommand{\vertiii}[1]{{\left\vert\kern-0.25ex\left\vert\kern-0.25ex\left\vert #1 
    \right\vert\kern-0.25ex\right\vert\kern-0.25ex\right\vert}}

\DeclareMathOperator{\argmin}{argmin}

\definecolor{darkgreen}{rgb}{0,0.7,0}

\newcommand{\iii}{{\vert\kern-0.25ex\vert\kern-0.25ex\vert}}

\definecolor{blue0}{RGB}{0,77,153} % dark blue
\definecolor{red0}{RGB}{179,0,77} % dark blue
\definecolor{green0}{RGB}{134,219,76} % dark blue
\definecolor{gray0}{RGB}{84,97,110}% gray blue

\newtheorem{theorem}{Theorem}

\newtheorem{lemma}[theorem]{Lemma}

\newtheorem{remark}[theorem]{Remark}
\theoremstyle{definition}
\numberwithin{equation}{section}
\numberwithin{theorem}{section}

\title{Mean-variance portfolio selection with tracking error penalization}

\author{William Lefebvre\thanks{BNP Paribas Global Markets, Universit\'e de Paris and Sorbonne Universit\'e, Laboratoire de Probabilit\'es, Statistique et Mod\'elisation (LPSM, UMR CNRS 8001), Building Sophie Germain, Avenue de France, 75013 Paris, \texttt{wlefebvre at lpsm.paris}} \quad Gr\'egoire Loeper\thanks{BNP Paribas Global Markets, School of Mathematics, Monash University, Clayton Campus, VIC, 3800, Australia, \texttt{gregoire.loeper at monash.edu}} \quad Huy\^en Pham \thanks{Universit\'e de Paris and Sorbonne Universit\'e, Laboratoire de Probabilit\'es, Statistique et Mod\'elisation (LPSM, UMR CNRS 8001), Building Sophie Germain, Avenue de France, 75013 Paris, \texttt{pham at lpsm.paris}}
}

\begin{document}
\maketitle
\begin{abstract}
    This paper studies a variation of the continuous-time mean-variance portfolio selection where a tracking-error penalization is added to the mean-variance criterion. The tracking error term penalizes the distance between the allocation controls and a refe\-rence portfolio with same wealth and fixed weights. Such consideration is motivated as fo\-llows: (i) On the one hand, it is a way to  robustify the mean-variance allocation in case of misspecified parameters, by ``fitting" it to a reference portfolio that can be agnostic to market parameters; (ii) On the other hand,  it is a procedure to track a benchmark and improve the Sharpe ratio of the resulting portfolio by considering a mean-variance criterion in the objective function. This problem is formulated as a McKean-Vlasov control problem. We provide explicit solutions for the optimal portfolio strategy and asymptotic expansions of the portfolio strategy and efficient frontier for small values of the tracking error parameter. Finally, we compare the Sharpe ratios obtained by the standard mean-variance allocation and the penalized one for four different reference portfolios: equal-weights, minimum-variance, equal risk contributions and shrinking portfolio.  
    %with weights equal to zero, which corresponds to a shrinking of the optimal controls. 
    This comparison is done on a simulated misspecified model,  
%with different misspecifications 
and on a backtest performed with historical data. Our results show that in most cases, the penalized portfolio outperforms in terms of Sharpe ratio both the standard mean-variance and the reference portfolio. 
\end{abstract}

%{\green commentaire WL} {\blue commentaire GL} {\red commentaire HP} 

\newpage

\section{Introduction}

The Markowitz mean-variance portfolio selection problem has been initially considered in \cite{markowitz1952portfolio} in a single-period model. In this framework, investement decision rules are made according to the objective of maximizing the expected return of the portfolio for a given financial risk quantified by its variance. The Markowitz portfolio is widely used in the financial industry due to its intuitive formulation and the fact that it produces, by construction, portfolios with high Sharpe ratios (defined as the ratio of the average of portfolio returns over their volatility), which is a key metric used to compare investment strategies. 

The mean-variance criterion involves the expected terminal wealth in a nonlinear way due to the presence of the variance term. In a continuous-time dynamic setting, this induces the so-called time inconsistency problem and prevents the direct use of the dynamic programming technique. A first approach, from \cite{zhou2000continuous}, consists in embedding the mean-variance problem into an auxiliary standard control problem that can be solved by using stochastic linear-quadratic theory. Some more recent approaches rely on the development of stochastic control techniques for control problems of McKean-Vlasov (MKV) type. MKV control problems are problems in which the equation of the state process and the cost function involve the law of this process and/or the law of the control, possibly in  a non-linear way. The mean-variance portfolio problem in continuous-time is a McKean-Vlasov control problem of the linear-quadratic type. The state diffusion, which represents the wealth of the portfolio, involves the state process and the control in a linear way while the cost involves the terminal value of the state and the square of its expectation due to the variance criterion. In \cite{andersson2011maximum}, the authors solved the mean-variance problem as a McKean-Vlasov control problem  by deriving a version of the Pontryagin maximum principle. More recently, \cite{pham2017dynamic} have developed a general dynamic programming approach for the control of MKV dynamics and applied it for the resolution of the mean-variance portfolio selection problem. In \cite{fischer2016continuous}, the mean-variance problem is viewed as the MKV limit of a family of controlled many-component weakly interacting systems. These prelimit problems are solved by standard dynamic programming, and the solution to the original problem is obtained by passage to the limit. 

A frequent criticism addressed to the mean-variance allocation 
%in the single-period and in the continuous-time setting 
is its sensitivity to the estimation of expected returns and covariance of the stocks and the risk of a poor out-of-sample performance. Several solutions to these issues have been considered. An approach consists in using a more sophisticated model than the Black-Scholes model, in which the parameters are stochastic or ambiguous and to take decisions under the worst-case scenario over all conceivable models.  
Robust mean-variance problems have thus  been considered in the economic and engineering literature, mostly on single-period or multi- period models; see, e.g., \cite{fabozzi2010robust}, \cite{pinar2016robust}, and \cite{liu2016correlation}. In a continuous-time setting,  \cite{ismail2019robust} have developed a robust approach by studying the mean-variance allocation with a market model where the model uncertainty affects the covariance matrix of multiple risky assets. In \cite{guo2020robust}, the authors study the problem of utility maximization under uncertain parameters in a model where the parameters of the model do not evolve freely within a given range, but are constrained via a penalty function. Let us also mention uncertain volatility models in \cite{matoussi2012robust} and \cite{lin2014optimal} for robust portfolio optimization with expected utility criterion.
Another approach is to rely on the shrinking of the portfolio weights or of the wealth invested in each risky asset in order to obtain a more sparse or more stable portfolio. In \cite{demiguel2009generalized}, the authors find single-period portfolios that perform well out-of-sample in the presence of estimation error. Their framework deals with the resolution of the traditional minimum-variance problem with the additional constraint that the norm of the portfolio-weight vector must be smaller than a given threshold. In \cite{ho2015weighted}, the authors study a one-period mean-variance problem in which the mean-variance objective function is regularized with a weighted elastic net penalty. They show that the use of this penalty can
be justified by a robust reformulation of the mean-variance criterion that directly accounts for parameter uncertainty. In the same spirit, in \cite{chen2013sparse}, $l_p$-norm regularized models are used to seek near-optimal sparse portfolios.

In this paper, we investigate the mean-variance portfolio selection in continuous time with a tracking error penalization. This penalization represents the distance  between the optimized portfolio composition and the composition of a reference portfolio with the same wealth but fixed weights that have been chosen in advance. Typical reference portfolios widely used in the financial industry are the equal weights, the minimum variance and the equal risk contribution (ERC) portfolios. The equal weights portfolio studied, e.g. in \cite{duchin2009markowitz}, is a portfolio where all the wealth of the investor is invested in risky assets and divided equally between the different assets. The minimum variance portfolio is a portfolio where all the wealth is invested in risky assets and portfolio weights are optimized in order to attain the minimal portfolio volatility. The ERC portfolio, presented in \cite{maillard2010properties} and in the monography \cite{roncalli2013introduction}, is totally invested in risky assets and optimized such that  the contributions of each asset to the total volatility of the portfolio are equal.  The mix of the mean-variance and of this tracking error criterion can be interpreted  in two different ways: (i) 
From a first viewpoint, it is a procedure  to regularize and robustify the mean-variance allocation. By choosing reference portfolio weights which are not based on the estimation of market parameters, or which are less sensible to estimation error, the allocation obtained is more robust to parameters estimation error than the standard mean-variance one. (ii) From a second viewpoint, this optimization permits  to mimic an allocation corresponding to the reference portfolio weights while improving  
its Sharpe ratio via the consideration  of the mean-variance criterion.

We tackle this problem as a McKean-Vlasov linear-quadratic control problem and adopt the approach developed in \cite{basei2019weak}, where the authors give a general method to solve this type of problems by means of a weak martingale optimality principle. We obtain explicit solutions for the optimal portfolio strategy and value function, and provide asymptotic expansions of the portfolio strategy and efficient frontier for small values of the portfolio tracking error penalization parameter. 
We then compare the Sharpe ratios obtained by the standard  mean-variance portfolio, the penalized one and the reference portfolio in two different ways. First, we compare these performances on simulated market data with misspecified market parameters. Different magnitudes of parameter misspecifications are used to illustrate the impact of the parameter estimation error on the performance of the different portfolios. In a second time, we compare the performances of these portfolios on a backtest based on historical market data. In these tests, we shall consider three reference portfolios cited above: the equal weights, the minimum variance and the equal risk contribution (ERC) portfolios. 
Finally, we will also consider the case where the reference portfolio weights are all equal to zero. This case corresponds to a shrinking of the wealth invested in the different risky assets along the investment horizon.

The rest of the paper is organized as follows. Section 2 formulates the mean-variance problem with tracking error. In Section 3 we derive explicit solutions for this control problem and provide expansion of this solution for small values of the tracking error penalization parameter. Section 4 is devoted to the applications of those results and to the comparison of the mean-variance, penalized and reference portfolio for the different reference portfolios presented above. We show the benefit of the penalized portfolio compared to the standard  mean-variance portfolio and the different reference portfolios on simulated and historical data in terms of Sharpe ratio and the lower sensitivity of the 
penalized portfolio to parameter estimation error. 
%\\
%For the proposed portfolios, they provide a moment- shrinkage interpretation and a Bayesian interpretation where the investor has a prior belief on portfolio weights rather than on moments of asset returns. Finally, they compare empirically the out-of-sample performance of the new portfolios they propose to 10 strategies in the literature across five data sets. They find that the norm-constrained portfolios often have a higher Sharpe ratio than the portfolio strategies in Jagannathan and Ma (2003), Ledoit and Wolf (2003, 2004), the 1/N portfolio, and other strategies in the literature, such as factor portfolios.\\
%The problem studied by these authors has no analytical solution, as the norm constraint on the vector of portfolio weights induce a non convexity of the feasible set of the problem. Instead of adding a constraint on the portfolio norm, we rather add a penalization on the vector of wealths invested in each stock in order to get analytical solutions and closed form results. We thus introduce this penalization as a running cost in the optimal control problem, which represents a euclidean distance between our portfolio composition and the composition of a benchmark, defined as a set of portfolio weights.\\

\section{Formulation of the problem}

Throughout this paper, we fix a finite horizon $T$ $\in$ $(0,\infty)$, and a complete probability space $\big( \Omega, \mathcal{F}, \mathbb{P}, 
\F= \left\{ \mathcal{F}_t \right\}_{0\leq t \leq T} \big)$ on which a standard $\F$-adapted $d$-dimensional Brownian motion $W$ $=$ $(W^1,...,W^d)$ is defined. We denote by $L_\F^2 (0,T;\mathbb{R}^d)$ the set of all $\R^d$-valued, measurable stochastic processes $(f_t)_{t\in [0,T]}$ adapted to $\F$ such that $\E\big[ \int_0^T |f_t|^2 dt \big] < \infty$.
We consider a financial market with price process  $P:=(P_t)_{t\in[0,T]}$, composed of one risk-free asset, assumed to be constant equal to one, i.e., $P^0$ $\equiv$ $1$, 
and $d$ risky assets on a finite investment horizon $[0,T]$. These assets price processes $P^i_t,\ i=1,...,d$ satisfy the following stochastic differential equation:

%{\green J'ai changé les notations et pris $b$ et $\sigma$ constants.}
%{\red HP: Ecrire the volatility matrix  $\sigma$ $=$ $(\sigma_{ij})_{i,j}$. Definir $\Sigma$ $=$ $\sigma\sigma^\top$ la matrice de covariance. Ecrire $W^j_t$ au lieu de $W^j(t)$, idem pour $W_t$}

\begin{equation}
\left\{ \begin{array}{ll}
dP^i_t \; = \; P^i_t\left(b_i\ dt+\sum_{j=1}^{n}\sigma_{ij}dW_t^{j}\right),\quad t\in[0,T]\\
P^i_0 \; > \; 0
\end{array}\right.
\end{equation}
where $b_{i}>0$ is the appreciation rate, and $\sigma:=(\sigma_{ij})_{i,j=1,...,d}\in \R^{d\times d}$ is the volatility matrix  of the $d$ stocks. 
We denote by $\Sigma := \sigma \sigma^\top$ the covariance matrix. 
%Define the covariance matrix
%\begin{equation}
%\sigma:=\begin{pmatrix}\sigma_{1}\\
%\vdots\\
%\sigma_{d}
%\end{pmatrix}:=(\sigma_{ij})_{d\times d}
%\end{equation}
Throughout this paper, we will assume that the following nondegeneracy condition holds
\begin{equation}
\Sigma \; \geq \; 
\delta \mathbb{I}_d, 
\end{equation}
for some $\delta>0$, where $\mathbb{I}_d$ is the $d\times d$ identity matrix.

Let us consider an investor with total wealth at time $t\geq 0$ denoted by $X_t$, starting from some initial capital $x_0$ $>$ $0$.  
%{\red Standard maintenant. Inutile de definir $N_t^i$. Definir directement comme dans mes papiers $\alpha_t$ comme le montant investi dans les actifs risques et ecrire la dynamique self-financing \eqref{eq:state_dynamic}.}\\
%{\green J'ai retiré cette portion.}
It is assumed that the trading of shares takes place continuously and transaction cost and consumptions are not considered. 
We define the set of admissible portfolio strategies $\alpha$  $=$ $(\alpha^1,\ldots,\alpha^d)$ as
\begin{equation}
\mathcal{A}:=\left\{ \alpha:\Omega\times[0,T]\rightarrow\mathbb{R}^{d}\ \mathrm{s.t}\ \alpha\ \mathrm{is}\ \F-\mathrm{adapted\ and}\ \int_{0}^{T}\mathbb{E}[|\alpha_{t}|^{2}]dt<\infty\right\},  
\end{equation}
where  $\alpha_t^i,\ i=1,...,d$ represents the total market value of the investor's  wealth invested in the $i$th asset at time $t$. The dynamics of the self-financed wealth process $X$ $=$ $X^\alpha$ associated to a 
portfolio strategy $\alpha$ $\in$ $\mathcal{A}$   is then  driven by
\begin{equation}
\label{eq:state_dynamic}
%\left\{ 
\begin{array}{ll}
dX_{t} \; = \;  \alpha_{t}^{\top} b\ dt+\alpha_{t}^{\top}\sigma dW_{t}. 
%\\ X_{0}>0
\end{array}
%\right.
\end{equation}

Given a risk aversion parameter $\mu>0$, and a reference weight $w_r$ $\in$ $\R^d$, 
the objective of the investor is to minimize over admissible portfolio strategies  a mean-variance functional to which is added a running cost:
\begin{equation}
\label{eq:cost}
J(\alpha)=\mu \textrm{Var}(X_T) - \E[X_T] + \E\Big[\int_{0}^{T}\left(\alpha_{t}-w_rX_{t}\right)^{\top}\Gamma\left(\alpha_{t}-w_rX_{t}\right)dt \Big]. 
\end{equation}
This running cost represents a running \textit{tracking error} between the portfolio composition $\alpha_t$ of the investor and the reference composition $w_r X_t$ of a portfolio of same wealth $X_t$ and constant weights $w_r$.
The matrix $\Gamma\in \R^{d\times d}$ is  symmetric positive definite and is used to introduce an anisotropy in the portfolio composition penalization.
%If this matrix is taken equal to the identity matrix, the term under the integral reduces to the euclidean distance between the composition of the investor's portfolio and that of the reference portfolio.
The penalization $\int_{0}^{T}\left(\alpha_{t}-w_rX_{t}\right)^{\top}\Gamma\left(\alpha_{t}-w_rX_{t}\right)$, which we will call  ``tracking error penalization",
is introduced in order to ensure that the portfolio of the investor does not move away
too much from this reference portfolio with respect to  the distance $|M| :=M^\top \Gamma M,\ M\in\R^d $.

The mean-variance portfolio selection with tracking error is then formulated as 
\begin{equation}
\label{control_problem}
    %\begin{cases}
        V_0 := \underset{\alpha \in \mathcal{A}}{\inf} J(\alpha),  %\\
    %    J(\alpha) &:=  \mu \textrm{Var}(x_T) - \E[x_t] + \E\left[\int_{0}^{T}\left(\alpha_{t}-w_{t}^{r}x_{t}\right)^{T}\Gamma\left(\alpha_{t}-w_{t}^{r}x_{t}\right)dt \right]
    %\end{cases}
\end{equation}
and an optimal allocation given the cost $J(\alpha)$ will be given by 

\begin{equation}
    \alpha_t^*  \; \in \;  \underset{\alpha\in \mathcal{A}}{\arg\min}\  J(\alpha). 
\end{equation}

We complete this section by recalling the solution to the mean-variance problem when there is no tracking error running cost, and which will serve later as benchmark for comparison when studying
the effect of the tracking error with several reference portfolios.

\begin{remark}[\textbf{Case of no tracking error}]
When $\Gamma =0$, it is known, see e.g. \cite{zhou2000continuous} that the optimal mean-variance strategy is given by 
\begin{equation}
\label{eq:optimal_control_classical}
    \alpha_t^* = \Sigma^{-1} b \left[ \frac{1}{2\mu}e^{b^\top \Sigma^{-1}b\ T} + x_0 - X_t^* \right],\quad 0\leq t\leq T,
\end{equation}
where $X_t^*$ is the wealth process associated to $\alpha^*$. The vector $\Sigma^{-1} b$, which depends only on the model parameters of the risky assets, determines the allocation in the risky assets.
\end{remark}

In the sequel, we study the quantitative impact of the tracking error running cost on the optimal mean-variance strategy.

\section{Solution allocation with tracking error}
\label{section:mv_tracking}
 
%{\red Cette hypothese est inutile: $\Gamma$ est deja supposee  positive definie, et des lors que $w_r$ n'est pas nul, on a $w_r^\top\Gamma w_r$ $>$ $0$. Mais cela n'inclut pas le shrinkingportfolio $w_r$ $=$ $0$!} 
%\begin{hyp}[H\ref{hyp:first}] \label{hyp:first}
%For the coefficients in \eqref{eq:cost}, there exists $\delta >0$ such that,
%    \begin{equation}
%        \Gamma\geq \delta \mathbb{I}_d, \quad \quad w_r^\top \Gamma w_r > 0.
%    \end{equation}
%\end{hyp}

Our main theoretical result provides an analytic characterization of the optimal control to the mean-variance problem with tracking error. 

\begin{theorem}
\label{theorem:mv_tracking}
There exist a unique pair $\left( K, \Lambda \right)\in C\left([0,T], \R_{+}^{*}\right) \times C\left([0,T], \R_{+}\right)$ solution to the system of ODEs
\begin{equation}
\label{eq:ode_system}
    \begin{cases}
dK_{t}=\left\{ \left(K_t b -\Gamma w_r \right)^\top S_t^{-1}\left(K_t b -\Gamma w_r \right)-w_r^{\top}\Gamma w_r \right\} dt, & K_{T}=\mu\\
\\
d\Lambda_{t}=\left\{ \left(\Lambda_t b -\Gamma w_r \right)^\top S_t^{-1}\left(\Lambda_t b -\Gamma w_r \right)-w_r^{\top}\Gamma w_r \right\} dt, & \Lambda_{T}=0\\
    \end{cases}
\end{equation}
where $S_t := K_t \Sigma + \Gamma$. The optimal control for problem  \eqref{control_problem} is then given by
\begin{equation}
\label{eq:optimal_control}
\alpha_{t}^{\Gamma}= S_t^{-1} \Gamma w_r X_t -S_{t}^{-1}b\big[K_{t}X_t +Y_{t}-(K_t - \Lambda_{t}) \E[X_{t}] \big],
\end{equation}
with 
\begin{align}
    \label{eq:closed_forms}
    &Y_t=-\frac{1}{2}e^{-\int_t^T b^\top S_s^{-1} (\Lambda_s b - \Gamma w_r)ds}\\
    &R_t=\frac{1}{2}\int_t^T b^\top S_s^{-1}b ~ e^{-2\int_s^T b^\top S_u^{-1} (\Lambda_u b - \Gamma w_r)du} ~ ds,
\end{align} 
and $X = X^{\alpha^\Gamma}$ is the wealth process associated to $\alpha^\Gamma$.  Moreover, we have 
\begin{equation}
    V_0 \; = \;  J(\alpha^\Gamma) \; = \; \Lambda_{0}X_{0}^{2}+2Y_{0}X_{0}+R_{0}. 
\end{equation}
\end{theorem}
\begin{proof}  Given the existence of a pair $\left( K, \Lambda \right)\in C\left([0,T], \R_{+}^{*}\right) \times C\left([0,T], \R_{+}\right)$ solution to \eqref{eq:ode_system},  the optimality of the control process 
in \eqref{eq:optimal_control} follows by the weak version of the martingale optimality principle as  developed in \cite{basei2019weak}. The arguments are recalled in appendix \ref{appendix:mv_tracking}.

Here, let us verify the existence and uniqueness of a solution to the system \eqref{eq:ode_system}. 
\begin{enumerate}[label=(\roman*)]
    \item We first consider the equation for $K$, which is a scalar Riccati equation. The equation for $K$ is associated to the standard linear-quadratic stochastic control problem:
    \begin{equation}
        \label{eq:K_control_prob}
        \tilde v(t,x) := \underset{\alpha \in \mathcal{A}}{\inf} \E \left[ \int_t^T \left( w_r^\top \Gamma  w_r (\tilde{X}_s^{t,x,\alpha})^2 -2 \alpha_s^\top \Gamma  w_r \tilde{X}_s^{t,x,\alpha} +\alpha_s^\top \Gamma  \alpha_s \right)ds \right]
    \end{equation}
    where $\tilde{X}_s^{t,x,\alpha}$ is the controlled linear dynamics solution to 
    \begin{equation}
        d\tilde{X}_s = \alpha_s^\top b\ ds + \alpha_s^\top \sigma dW_s,\quad t\leq s\leq T,\ \tilde{X}_t = x.
    \end{equation}
    By a standard result in control theory \cite[Ch. 6, Thm. 6.1, 7.1, 7.2]{yong1999stochastic}, there exists a unique solution $K\in C([0,T], \R_+)$ to the first equation of system \eqref{eq:ode_system} (more, $K\in C([0,T], \R_+^*)$ if $w_r$ is nonzero). In this case, we have $\tilde v(t,x)=x^\top K_t x$. 
    \item Given $K$, we consider the equation for $\Lambda$. This is also a scalar Riccati equation. By the same arguments as for the $K$ equation, there exists a unique solution $\Lambda \in C([0,T], \R_+)$ to the second equation of \eqref{eq:ode_system}, provided that
    \begin{equation}
    \Lambda_T \geq 0,\quad \quad w_r^\top \Gamma w_r - w_r^\top \Gamma \left( K_t \Sigma + \Gamma \right)^{-1} \Gamma w_r \geq 0, \quad \quad K_t \Sigma + \Gamma \geq \delta \mathbb{I}_d,\quad \quad 0\leq t\leq T 
    \end{equation}
    for some $\delta>0$. We already have that $\Lambda_T = 0$. From the fact that $K> 0$, together with the  nondegeneracy condition on the matrix $\Sigma$, we have that $K_t \Sigma +\Gamma \geq \Gamma \geq \delta \mathbb{I}_d$. 
    Since $\Gamma$ $>$ $0$, and under the nondegeneracy condition of matrix $\Sigma$, we can use the Woodbury matrix identity to obtain 
    \begin{equation}
        \left( K_t \Sigma + \Gamma \right)^{-1} = \Gamma^{-1} - \Gamma^{-1} \left( \Gamma^{-1} + \frac{\Sigma^{-1}}{K_t} \right)^{-1} \Gamma^{-1}.
    \end{equation}
    We then get 
    \begin{equation}
        w_r^\top \Gamma w_r - w_r^\top \Gamma \left( K_t \Sigma + \Gamma \right)^{-1} \Gamma w_r = w_r^\top \left( \Gamma^{-1} + \frac{\Sigma^{-1}}{K_t} \right)^{-1} w_r \geq 0.
    \end{equation}

    \item Given $(K, \Lambda)$, the equation for $Y$ is a linear ODE, whose unique continuous solution is explicitly given by
    \begin{equation}
        Y_t = -\frac{1}{2}e^{-\int_t^T b^\top S_s^{-1} (\Lambda_s b - \Gamma w_r)ds}. 
    \end{equation}
    \item Given $(K, \Lambda, Y)$, $R$ can be directly integrated into
    \begin{equation}
        R_t = \frac{1}{2}\int_t^T b^\top S_s^{-1}b ~ e^{-2\int_s^T b^\top S_u^{-1} (\Lambda_u b - \Gamma w_r)du} ~ ds. 
    \end{equation}
\end{enumerate}
\end{proof}

\vspace{3mm}

We can see from the expression of the optimal control \eqref{eq:optimal_control} that the allocation in the risky assets has two components. One component is determined by the vector $S_t^{-1}\Gamma w_r = \left( K_t \Sigma + \Gamma \right)^{-1}\Gamma w_r$ with leverage $X_t$, and the second one by the vector $S_t^{-1}b = \left( K_t \Sigma + \Gamma \right)^{-1} b$ with leverage $\left[K_{t}X_t +Y_{t}-(K_t - \Lambda_{t}) \E[X_{t}]\right]$. Computing the average wealth $\overline{X}$ $=$ $\E[X]$  
associated to $\alpha
^\Gamma$, we can express the control $\alpha^\Gamma$ as a function of the initial wealth of the investor $x_0$ and the current wealth $X_t$

\begin{align}
\label{eq:optimal_control_x0}
    \alpha_t^\Gamma =& \;  S_t^{-1} \Gamma w_r X_t-\Lambda_t S_t^{-1}b \left(X_{0}C_{0,t} + \frac{1}{2}H_t \right)\\
    &\; +S_{t}^{-1}b\left[K_{t}\left(X_{0}C_{0,t}+ \frac{1}{2}H_t -X_t\right) -Y_{t}\right]
\end{align}
where we set $C_{s, t}:= e^{-\int_{s}^{t}b^{\top}S_{u}^{-1}\left(\Lambda_{u}b-\Gamma w_r\right)du}$ and $H_t :=C_{t,T}\int_{0}^{t} C_{s,t}^2\ b^\top S_s^{-1} b \ ds$.

\vspace{1mm}

\begin{remark}
In the case when $\Gamma$ is the null matrix,  $\Gamma$ $=$ $\mathbf{0}$,  we see that the first component of the optimal control \eqref{eq:optimal_control_x0} vanishes,
\begin{equation}
    Y_t = -\frac{1}{2},\quad R_t = \frac{1}{4\mu}\left( 1-e^{b^\top \Sigma^{-1} b\ (T-t)} \right),
\end{equation}
and the system of ODES \eqref{eq:ode_system} of $(K, \Lambda)$ becomes 
\begin{equation}
\label{eq:ode_system_classic}
    \begin{cases}
dK_{t}= K_t b^\top \Sigma^{-1} b\ dt, & K_{T}=\mu\\
\\
d\Lambda_{t}=\frac{\Lambda_t^2}{K_t} b^\top \Sigma^{-1} b\ dt, & \Lambda_{T}=0,
    \end{cases}
\end{equation}
which yields the explicit forms 
\begin{equation}
    K_t = \mu e^{- b^\top \Sigma^{-1} b\ (T-t)}, \quad \Lambda_t = 0.
\end{equation}
We get $S_t^{-1}=\frac{\Sigma^{-1}}{K_t}=\frac{\Sigma^{-1}e^{b^\top \Sigma^{-1} b\ (T-t)}}{\mu}$, $C_{\cdot, \cdot}=1$ and $H_t=\frac{1}{\mu}\int_{0}^{t}b^\top \Sigma^{-1} b\ e^{b^\top \Sigma^{-1} b\ (T-s)} ds$. The first line of the optimal control $\alpha^\Gamma$ equation vanishes and the second line can be rewritten as 
\begin{align}
    \alpha_t^\Gamma = \Sigma^{-1}b\left[\frac{1}{2\mu}\left(e^{b^\top \Sigma^{-1} b\ (T-t)} + \int_{0}^{t}b^\top \Sigma^{-1} b\ e^{b^\top \Sigma^{-1} b\ (T-s)} ds \right) + X_{0}-X_t \right].
    %\alpha_t^\Gamma =& \Sigma^{-1}b \left( \frac{1}{2\mu}e^{b^\top \Sigma^{-1} b\ (T-t)} - 1 \right)\\
    %&+\Sigma^{-1}b X_0\\
    %&+\Sigma^{-1}b \frac{1}{2\mu}\int_0^t b^\top \Sigma^{-1}b e^{b^\top \Sigma^{-1} b\ (T-s)} \ ds 
\end{align}
Computing the integral in this expression, we recover the optimal control of the classical mean-variance problem \eqref{eq:optimal_control_classical}.
\end{remark}

%{\red  Discuter l'existence de ce systeme d'ODE comme cela est fait dans l'article Basei-Pham. L'existence decoule soit de resultats classiques sur les equations de Riccati, soit en liant ce systeme d'ODEa un probleme lineaire-qaudratique, cf references dans mon article avec Basei. Retravailler les expressions de $K,\Lambda,Y$. $\Lambda$ est explicite quand on connait $K$, de meme pour $Y$ puis $R$. On doit  pouvoir etudier la monotonie de $K$ $=$ $K^\Gamma$ en fonction de $\Gamma$.}\\
%{\green J'ai rajouté une partie sur l'existence et l'unicité du système d'ODE dans la preuve en appendice. J'ai donné les expressions explicites de $Y$ et $R$ en fonction de $K$ et $R$. Je n'ai pas encore trouvé de solution explicite pour $\Lambda$ en fonction de $K$.}

\begin{remark}[\textbf{Limit of $\alpha_t^\gamma$ for $\Gamma =\gamma\mathbb{I}_d \rightarrow \infty$}]
If we consider $\Gamma$ in the form  $\Gamma=\gamma \mathbb{I}_d$, the optimal control can be rewritten as 
\begin{equation}
\label{eq:control_expression_limit}
    \alpha_{t}^{\gamma}= \left( \mathbb{I}_d + \frac{K_t}{\gamma} \Sigma \right)^{-1} w_r X_t - \frac{1}{\gamma} \left( \mathbb{I}_d + \frac{K_t}{\gamma} \Sigma \right)^{-1}b\left[K_{t}X_t +Y_{t}-(K_t - \Lambda_{t})\overline{X}_{t}\right].
\end{equation}
We show in appendix \ref{appendix:bounded_quotients} that $K_t$ and $\Lambda_t$ are bounded functions of the penalization parameter $\gamma$, thus $\frac{K_t}{\gamma},~ \frac{\Lambda_t}{\gamma}\underset{\gamma\rightarrow \infty}{\longrightarrow}0$. \\
We rewrite $Y_t$ as 
\begin{equation}
        Y_t = -\frac{1}{2}e^{b^\top w_r (T-t)} e^{ -\int_t^T \frac{1}{\gamma} b^\top \left( \mathbb{I}_d+\frac{K_s}{\gamma}\Sigma \right)^{-1}\left( \Lambda_s b +K_s \Sigma w_r \right)ds }
    \end{equation}
    and we get that $Y_t \underset{\gamma\rightarrow\infty}{\longrightarrow}-\frac{1}{2}e^{b^\top w_r (T-t)}$. Thus the second term of \eqref{eq:control_expression_limit} vanishes and we get
    \begin{equation}
        \alpha_t^\gamma \underset{\gamma\rightarrow\infty}{\longrightarrow} w_r X_t
    \end{equation}
   which corresponds to the reference portfolio.
\end{remark}

\begin{remark}[\textbf{Expansion for $\Gamma =\gamma\mathbb{I}_d \rightarrow 0$}]
We take $\Gamma = \gamma \mathbb{I}_d$. Since the covariance matrix $\Sigma$ is symmetric,  there exists an invertible matrix $Q\in \R^{d\times d}$ and a diagonal matrix $D\in \R^{d\times d}$ such that $\Sigma=Q\cdot D\cdot Q^{-1}$. We can then rewrite the matrix $S_t^{-1}:=\left( K_t \Sigma +\gamma\mathbb{I}_d \right)^{-1}$ as 
\begin{equation}
    S_t^{-1} = Q\cdot \left( K_t D +\gamma \mathbb{I}_d \right)^{-1} Q^{-1}
\end{equation}
with
\begin{equation}
\left(\left(K_t D+\gamma \mathbb{I}_d\right)^{-1}\right)_{ij} =
\begin{cases}
\frac{1}{K_t d_i + \gamma}\ &\textrm{if}\ i=j \\
0\ &if\ i\neq j
\end{cases}
\end{equation}
where $d_i$ is the $i$-th diagonal value of the diagonal matrix $D$. From the nondegeneracy condition of the covariance matrix, we have $d_i>0,\ \forall i\in \llbracket 1, n \rrbracket$. 
As $\gamma \longrightarrow 0$, we want to write the Taylor expansion of the diagonal elements of the inverse matrix $\left(D+\gamma \mathbb{I}_d\right)^{-1}$ equal to $\frac{1}{K_t d_i}\left( 1 + \frac{\gamma}{K_t d_i} \right)^{-1}$.
We have that $K_t \underset{\gamma\rightarrow 0}{\longrightarrow}\mu e^{- \rho(T-t)}$, thus $\frac{\gamma}{K_t}\underset{\gamma\rightarrow 0}{\longrightarrow}0$. 
We can then write the Taylor expansion of the matrix $S_t^{-1}$ as 
\begin{equation}
    S_t^{-1}=\frac{\Sigma^{-1}}{K_t}- \gamma \frac{\left( \Sigma^{-1} \right)^2}{K_t^2} + O(\gamma^2)
\end{equation}
keeping only the terms up to the linear term in $\gamma$. \\
Putting this expression in the differential equation of $K$, and keeping only the terms up to the linear term in $\gamma$, we get the differential equation
\begin{equation}
    \label{eq:K_ODE_expansion}
    \frac{dK_t}{dt}= K_t \rho -\gamma \| w_r + \Sigma^{-1}b \|^2 + O(\gamma^2),
\end{equation}
where we set $\rho:=b^\top \Sigma^{-1}b$. We look for a solution to this equation of the form
\begin{equation}
    K_t^\gamma = K_t^0 + \gamma K_t^1 + O(\gamma^2).
\end{equation}
Putting this expression in the differential equation \eqref{eq:K_ODE_expansion}, we get two differential equations, for the leading order and the linear order in $\gamma$ respectively
\begin{equation}
    \begin{cases}
        \frac{dK_t^0}{dt}=K_t^0 \rho,\quad &K_T^0 = \mu\\ 
        \frac{dK_t^1}{dt}=K_t^1 \rho - \| w_r + \Sigma^{-1}b \|^2, \quad &K_T^1 = 0
    \end{cases}
\end{equation}
which yield the explicit solution 
\begin{equation}
    K_t^\gamma = K_t^0 + \gamma \| w_r + \Sigma^{-1}b \|^2 ~ \frac{1-e^{-\rho (T-t)}}{\rho} + O(\gamma^2)
\end{equation}
where $K_t^0 = \mu e^{- \rho(T-t)}$ is the solution to the differential equation in the unpenalized case.\\
From the expansion for $K$, we can write the expansion of the differential equation for $\Lambda$ up to the linear term in $\gamma$. 
We use the expansion 
\begin{equation}
    \frac{1}{K_t^\gamma}=\frac{1}{K_t^0}\left( 1 - \gamma \| w_r + \Sigma^{-1}b \|^2 ~ \frac{1-e^{-\rho (T-t)}}{K_t^0 \rho} \right) +O(\gamma^2)
\end{equation}
and we get the following expansion of the differential equation of $\Lambda$
\begin{align}
    \label{eq:lamb_ODE_expansion}
    \frac{d\Lambda_t}{dt}=& \frac{\Lambda_t^2}{K_t^0} \rho \left( 1 - \gamma \| w_r + \Sigma^{-1}b \|^2 ~ \frac{1-e^{-\rho (T-t)}}{K_t^0 \rho} \right)\\
    &-\gamma \left( 2 \frac{\Lambda_t}{K_t^0}b^\top \Sigma^{-1}w_r - \left(\frac{\Lambda_t}{K_t^0}\right)^2 b^\top \Sigma^{-2}b - \| w_r \|^2 \right) + O(\gamma^2).
\end{align}
As before, we look for a solution of this differential equation of the form
\begin{equation}
    \Lambda_t^\gamma = \Lambda_t^0 + \gamma \Lambda_t^1 + O(\gamma^2).
\end{equation}
Plugging this expression into  the equation \eqref{eq:lamb_ODE_expansion}, we get the two following differential equations
\begin{equation}
    \begin{cases}
    \frac{d\Lambda_t^0}{dt}= \frac{\left(\Lambda_t^0\right)^2}{K_t^0} \rho,\quad &\Lambda_T^0 = 0\\ \\
    \frac{d\Lambda_t^1}{dt}= 2\frac{\Lambda_t^0 \Lambda_t^1}{K_t^0}\rho -\left(\frac{\Lambda_t^0}{K_t^0}\right)^2 \rho \| w_r + \Sigma^{-1}b \|^2 ~ \frac{1-e^{-\rho (T-t)}}{\rho} \\
    - \left( 2 \frac{\Lambda_t^0}{K_t^0}b^\top \Sigma^{-1}w_r + \left(\frac{\Lambda_t^0}{K_t^0}\right)^2 b^\top \Sigma^{-2}b + \| w_r \|^2 \right), \quad &\Lambda_T^1=0. 
    %& \Lambda_T^0 = 0,\ \Lambda_T^1=0 
    \end{cases}
\end{equation}
The first differential equation yields the solution $\Lambda_t^0 = 0,\ \forall t\in [0,T]$. Replacing $\Lambda_t^0$ by this value in the second differential equation, we get the equation
\begin{equation}
    \frac{d\Lambda_t^1}{dt}=-\| w_r \|^2
\end{equation}
and obtain the solution
\begin{equation}
    \Lambda_t^\gamma = \gamma \| w_r \|^2 (T-t) + O(\gamma^2).
\end{equation}
We can also compute the first order expansion of $C_{\cdot,\cdot}$
\begin{align}
    C_{s,t}^\gamma=& 1-\gamma \int_s^t \frac{\rho}{K_u^0}\left( \| w_r \|^2 (T-u) - \frac{b^\top \Sigma^{-1}w_r}{\rho} \right)du + O(\gamma^2)\\
     =& 1-\gamma C_{s,t}^1 +O(\gamma^2)
\end{align}
where we set 
\begin{equation}
    C_{s,t}^1 := \frac{e^{\rho(T-t)}}{\mu\rho}\left\{ \rho \|w_r\|^2(t-s)  + \left( e^{\rho(t-s)}-1 \right)\left( \|w_r\|^2 (\rho T-1) - b^\top \Sigma^{-1}w_r \right) \right\},
\end{equation}
 and we have
\begin{equation}
    Y_t^\gamma = -\frac{1}{2} +\frac{\gamma}{2} C_{t,T}^1.
\end{equation}
The last expansion we need to compute before rewritting the optimal control is the expansion of $H_t$. We can rewrite 
\begin{equation}
    H_t = \frac{e^{\rho T}}{\mu}\left( 1-e^{-\rho t} \right) - \gamma H_t^1 +O(\gamma^2)
\end{equation}
with
\begin{equation}
    H_t^1 := \int_0^t \left( 2C_{s,t}^1 +C_{t,T}^1 \right) \frac{b^\top \Sigma^{-1}b}{K_s^0} ds + \int_0^t b^\top \frac{\Sigma^{-1}}{\left(K_s^0 \right)^2} \left( K_s^1 \mathbb{I}_d + \Sigma^{-1} \right)b~ ds.
\end{equation}
As shown in appendix \ref{appendix:control_DL}, we can rewrite the optimal control
\begin{equation}
    \label{eq:optimal_control_expansion}
    \alpha_{t}^{\gamma}= \Sigma^{-1}b~ \alpha_t^0 +\gamma \left( \Sigma^{-1}w_r~ \alpha_t^{1,3} - \Sigma^{-2}b~ \alpha_t^{1,2} - \Sigma^{-1}b~ \alpha_t^{1,1} \right) + O(\gamma^2)
\end{equation}
where we set $\Sigma^{-2}$ $:=$ $(\Sigma^{-1})^2$, and with 
\begin{equation}
    \label{eq:control_expansion_weights}
    \begin{cases}
        \alpha_t^0 = \frac{1}{2\mu} e^{\rho T} + X_0 - X_t \\
        \alpha_t^{1,1} = \frac{\|w_r\|^2}{K_t^0} (T-t)\left( X_0 + \frac{e^{\rho T}}{\mu}\left( 1-e^{-\rho t} \right) \right) + X_0 C_{0,t}^1 + \frac{H_t^1}{2} + \frac{K_t^1}{2\left( K_t^0 \right)^2} + \frac{C_{t,T}}{2K_t^0} \\
        \alpha_t^{1,2} = \frac{e^{\rho T}}{2K_t^0 \mu} \left( 1-e^{-\rho t} \right) + \frac{1}{2\left( K_t^0 \right)^2} \\
        \alpha_t^{1,3} = \frac{X_t}{K_t^0}.
    \end{cases}
\end{equation}
We see that for $\gamma=0$, we recover the classical mean-variance optimal control. For non-zero values of $\gamma$, we see that a mix of three different portfolio allocations is obtained. The weight of the allocation $\Sigma^{-1}b$ is modified and two allocations $\Sigma^{-2}b$ and $\Sigma^{-1}w_r$ appear with weights $\gamma \alpha_t^{1,2}$ and $\gamma \alpha_t^{1,3}$.

From  this expansion of the control $\alpha^\gamma$, we can compute the first order asymptotic expansion in $\gamma$ of the equation giving the relation between the variance of the terminal wealth of the portfolio and its expectation. In the classical mean-variance case, this equation is called the \textit{efficient frontier} formula. As shown in appendix \ref{appendix:var_computation}, with the tracking error penalization, the first order asymptotic expansion in $\gamma$ gives 
\begin{align}
    Var(X_T) =& \frac{e^{-\rho T}}{1-e^{-\rho T}} \left( \overline{X_T}^0 - X_0 \right)^2 \\ 
    &+ \gamma \left\{ \frac{b^\top \Sigma^{-1}w_r}{\mu^2}\left[ X_0 T - \frac{1}{2\mu}e^{\rho T} \left(T - \frac{1 -e^{-\rho T} }{\rho}\right) \right]\right.\\
    &\left.- \int_0^T \left( \frac{\rho}{\mu}\alpha_s^{1,1} + \frac{b^\top \Sigma^{-2} b}{\mu} \alpha_s^{1,2} \right)e^{-\rho (T-s)}ds\right\} + O(\gamma^2).
\end{align}
The leading order term corresponds to the efficient frontier equation of the classical mean-variance allocation computed in \cite{zhou2000continuous}, and thus  
for $\gamma=0$, we  recover this classical result. The linear term in $\gamma$ contains contributions of the three perturbative allocations. A modification of "leverage" of the original mean-variance allocation $\Sigma^{-1}b$ and two different allocations $\Sigma^{-2}b$ and $\Sigma^{-1}w_r$.
\end{remark}

\section{Applications and numerical results}

In this section, we  apply the results of the previous section and study the allocation obtained by considering  four different static portfolios as reference. First, we shall study these allocations on simulated data, in the case of misspecified parameters. The misspecification of parameters means that the market parameters used to compute the portfolio allocations are different from  the ones driving the stocks prices. This study allows us to estimate the impact of the  estimation error  on the portfolio performance. In a second time, we perform a backtest and run the different portfolios on real market data. To simplify the presentation, we will assume now that the tracking error penalization matrix is in the form  $\Gamma = \gamma \mathbb{I}_d$ with 
$\gamma\in \R_+^*$. With this simplification, we have $S_t^{-1}=\left( K_t \Sigma + \gamma \mathbb{I}_d \right)^{-1}$ and we can rewrite the system of ODEs \eqref{eq:ode_system} and the optimal control \eqref{eq:optimal_control} as

\begin{equation}
\label{eq:ode_system_simp}
    \begin{cases}
dK_{t}=\left\{ \left(K_t b - \gamma w_r \right)^\top S_t^{-1}\left(K_t b -\gamma w_r \right)-\gamma \left(w_r\right)^{\top}w_r \right\} dt, & K_{T}=\mu\\
\\
d\Lambda_{t}=\left\{ \left(\Lambda_t b -\gamma w_r \right)^\top S_t^{-1}\left(\Lambda_t b -\gamma w_r \right)-\gamma\left(w_r\right)^{\top} w_r \right\} dt, & \Lambda_{T}=0
    \end{cases}
\end{equation}
and
\begin{align}
\label{eq:optimal_control_simp}
    \alpha_t^\gamma =& \gamma S_t^{-1} w_r X_t-\Lambda_t S_t^{-1}b \left(X_{0}C_{0,t} + \frac{1}{2}H_t \right)\\
    &+S_{t}^{-1}b\left[K_{t}\left(X_{0}C_{0,t}+ \frac{1}{2}H_t -X_t\right) -Y_{t}\right]
\end{align}
where
\begin{equation}
    S_t = K_t\Sigma + \gamma \mathbb{I}_d,\quad C_{s, t}:= e^{-\int_{s}^{t}b^{\top}S_{u}^{-1}\left(\Lambda_{u}b-\gamma w_r\right)du},\quad Y_t=-\frac{1}{2}C_{t, T}.
\end{equation}

We will consider three different classical allocations as reference portfolio.
\begin{enumerate}
    \item \textbf{Equal-weights portfolio}: in this classical equal-weights portfolio, the same capital  is invested in each asset, thus
    \begin{equation}
        \label{eq:ew_weights}
        w_r^{\textrm{ew}} = \frac{1}{d}~e
    \end{equation} where $d$ is the number of risky assets considered and $e\in\R^d$ is the vector of ones.
    \item \textbf{Minimum variance portfolio}: the minimum variance portfolio is the portfolio which achieves the lowest variance while investing all its wealth in the risky assets. The weight vector of this portfolio is equal to
    \begin{equation}
    \label{eq:min_var_weights}
        w_r^{\textrm{min-var}} = \frac{\Sigma^{-1}e}{e^\top\Sigma^{-1}e}.
    \end{equation}
    These weights correspond to the one-period Markowitz portfolio when every asset expected return $b_i$ is taken equal to 1. In that case, only the portfolio variance is relevant and is minimized during the optimization process. \\
    \item \textbf{ERC portfolio}: the equal risk contributions (ERC) portfolio, presented in \cite{maillard2010properties} and in the monograph \cite{roncalli2013introduction} is constructed by choosing a risk measure and computing the risk contribution of each asset to the global risk of the portfolio. When the portfolio volatility is chosen as the risk measure, the principle of the ERC portfolio lays in the fact that the volatility function satisfies the hypothesis of Euler's theorem and can be reduced to the sum of its arguments multiplied by their first partial derivatives. The portfolio volatility $\sigma(w)=\sqrt{w
   ^\top \Sigma w}$ of a portfolio with weights vector $w\in\R
   ^d$ can then be rewritten as
    \begin{equation}
        \sigma(w)=\sum_{i=1}^d w^i \partial_i \sigma(w) = \sum_{i=1}^d \frac{w^i \left( \Sigma w \right)^i}{\sigma(w)}.
    \end{equation}
    %The risk contribution $\textrm{RC}_i$ of the $i$-th asset {\red pas clair cette phrase}  the global risk of the portfolio can then be defined as $\textrm{RC}_i := \frac{w^i \left( \Sigma w \right)^i}{\sigma(w)}$, and the equal risk contribution is attained when the portfolio weights $w^*$ are given by
    The term under the sum $\frac{w^i \left( \Sigma w \right)^i}{\sigma(w)}$, corresponding to the $i$-th asset, can be interpreted as the contribution of this risky asset to the total portfolio volatility. The equal risk contribution allocation is then defined as the allocation in which these contributions are equal for all the risky assets of the portfolio, $\frac{w^i \left( \Sigma w \right)^i}{\sigma(w)} = \frac{w^j \left( \Sigma w \right)^j}{\sigma(w)}$ for every $i,j\in \llbracket 1,d \rrbracket$. The equal risk contribution allocation is thus obtained when the portfolio weights $w^*$ are given by
    \begin{equation}
        w^* = \left\{ w\in [0,1]^d: \ \sum_{i=1}^d w^i =1,\ w^i \left( \Sigma w \right)^i = w^j \left( \Sigma w \right)^j,\ \forall i,j\in \llbracket 1, d\rrbracket \right\}.
    \end{equation}
    With this risk measure, the ERC portfolio weights can be expressed in a closed-form only in the case where the correlations between every couple of stocks are equal, that is $\textrm{corr}(P_i, P_j)=c,\ \forall
   ~i,j\in \llbracket1, d\rrbracket $, with the additional assumption that $c\geq -\frac{1}{d-1}$. Under these assumptions, and with the constaint that $\sum_{i=1}^d \left(w_r^{\textrm{erc}}\right)_i =1$, the weights of this portfolio are equal to
   \begin{equation}
        \label{eq:erc_weights}
       \left(w_r^{\textrm{erc}}\right)_i=\frac{\sigma_i^{-1}}{\sum_{j=1}^d \sigma_j^{-1}}
   \end{equation}
   where $\sigma_i$ is the volatility of the $i$-th asset.
   
   In the general case, the weights of the ERC portfolio do not  have a closed form and must be computed numerically by solving the following optimization problem
   \begin{align}
       w_r^{\textrm{erc}}=\underset{w\in\R^d}{\argmin} \sum_{i=1}^d \sum_{j=1}^d \left( w^i \left(\Sigma w\right)^i - w^j \left(\Sigma w\right)^j \right)^2\\
       \textrm{s.t}\ e^\top w =1\ \textrm{and}\ 0\leq w^i \leq 1,\ \forall i\in \llbracket 1,d\rrbracket.
   \end{align}
   \item \textbf{Control shrinking (zero portfolio)}: this is the portfolio where all weights are equal to zero, $w_r^i =0$ for all $i$. This case corresponds to a shrinking of the controls of the penalized allocation, in the same spirit as the shrinking of regression coefficients in the Ridge regression (or Tikhonov regularization).
\end{enumerate}

\subsection{Performance comparison with Monte Carlo simulations}

In this section we compare, for each reference portfolio, the classical dynamic mean-variance allocation, the reference portfolio and the  ``tracking error" penalized portfolio. In a real investment situation, expected return and covariance estimates are noisy and biased. Thus, in order to compare the three portfolios and observe the impact of adding a tracking error penalization in the mean-variance allocation, we will run Monte Carlo simulations, assuming that the real-world expected returns $b_{\textrm{real}}$ and covariances $\sigma_{\textrm{real}}$ are equal to reference expected returns $b_0$ and covariances $\sigma_0$ plus some noise:
%We compare these portfolio on simulated market data for well specified and misspecified market parameters. In the well specified case, the expected returns and covariance matrix of the assets are the same as the ones used to compute the three allocations. In the misspecified case, these parameters are different. \\
%For the computation of the allocations, we consider, as example, four assets with expected returns $b_{alloc}$, volatilities $v_{alloc}$ and correlations $C_{alloc}$ equal to
\begin{align}
    b_{0} = 
        \begin{pmatrix}
            0.12 \\
            0.14 \\
            0.16 \\
            0.10
        \end{pmatrix},
    \quad
    v_{0} = 
        \begin{pmatrix}
            0.20 \\
            0.30 \\
            0.40 \\
            0.50
        \end{pmatrix},
    \quad
    C_{0} =
        \begin{pmatrix}
            1. & 0.05 & -0.05 & 0.10 \\
            0.05 & 1. & -0.03 & 0.12 \\
            -0.05 & -0.03 & 1. & -0.13 \\
            0.10 & 0.12 & -0.13 & 1. 
        \end{pmatrix},
 \end{align}
 with the volatilties $v_0$ and correlations $C_0$ and 
 \begin{equation}
     b_{\textrm{real}}= b_0 + \epsilon \times \textrm{noise},\quad \sigma_{\textrm{real}}= \sigma_0 + \epsilon \times \textrm{noise}
 \end{equation}
 where the covariance matrix $\sigma_0$ is obtained from $v_0$ and $C_0$.
    The noise follows a standard normal distribution $\mathcal{N}(0,1)$ and $\epsilon$ is its magnitude. We use Monte Carlo simulations to estimate the expected Sharpe ratio of each portfolio, equal to the average of the portfolio daily returns $R$ divided by the standard deviation of those returns: 
%\begin{equation}
$\E \big[\frac{\E[R]}{\textrm{Stdev}(R)}\big]$. 
%\end{equation}
 %its tracking error 
 %\begin{equation}
 %    \E \left[ \int_0^T \|\alpha_t-w_r X_t\|^2 \right]
 %\end{equation}
%and expected cumulative turnover
%\begin{equation}
%    \E \left[\sum_{i=1}^N  \left( |\alpha_i^0 - \alpha_{i-1}^{0}| + \sum_{j=1}^d|\alpha_i^j - \alpha_{i-1}^{j}|\right)\right]
%\end{equation}
%where the $i$ index represents the trading days, $\alpha_i^0$ is the wealth kept as cash at day $i$ and $\alpha_i^j$ is the wealth invested in the risky asset $j$ on day $i$.

\vspace{1mm}

 We consider an investment horizon of one year, with 252 business days and a daily rebalancing of the portfolio. The risk aversion parameter $\mu$ is chosen so that the targeted annual return of the classical mean-variance allocation is equal to $20\%$, thus $\mu = \frac{e^{b^{\top} \Sigma^{-1} b}}{2x_0 * 1.20}$ according to \cite{zhou2000continuous}. The initial wealth of the investor $x_0$ is chosen equal to 1 and we choose the penalization parameter $\gamma= \mu/100$. Indeed, as the value of $\mu$ depends on the value of the stocks expected return and covariance  matrix  and  on  the  targeted  return,  and  can  be  very  big,  we  express $\gamma$  a function of this $\mu$ in order for the penalization to be relevant and non-negligible.
 
 For each reference portfolio, we compare the reference portfolio, the classical mean-variance allocation and the penalized one for values of noise amplitude $\epsilon$ ranging from 0 to 1. For each value of $\epsilon$, we run 2000 scenarios and we plot the graphs of the average Sharpe ratio as a function of $\epsilon$.
 
 On the following graphs, we can see that in the four cases, the mean-variance and the penalized portfolios are superior to the reference. In the case where the equal weights portfolio is chosen as reference, the penalized portfolio's Sharpe ratio is lower than the mean-variance one for small values of $\epsilon$. For $\epsilon$ greater than approximately 0.25, the penalized portfolio's Sharpe ratio becomes larger and the gap with the mean-variance's Sharpe tends to increase with $\epsilon$. The same phenomenon occurs in the case where the ERC portfolio is chosen as reference, with a smaller gap between the mean-variance and penalized portfolios' Sharpe ratios. When the minimum variance portfolio is chosen as reference, the penalized portfolio's Sharpe ratio is lower than the one of the mean-variance portfolio for all $\epsilon$ in the interval $[0,1]$. This is certainly due to the sensitivity of the minimum variance portfolio to the estimator of the covariance matrix. Finally, in the case of the control shrinking, the Sharpe ratio of the penalized portfolio is significantly  higher that the Sharpe ratio of the mean-variance portfolio, for every value of the noise amplitude $\epsilon$ in the interval $[0,1]$.

\clearpage

\begin{itemize}
\item Equal-weights reference portfolio
\begin{figure}[h!]
\centering
  \includegraphics[width=9cm,height=7cm]{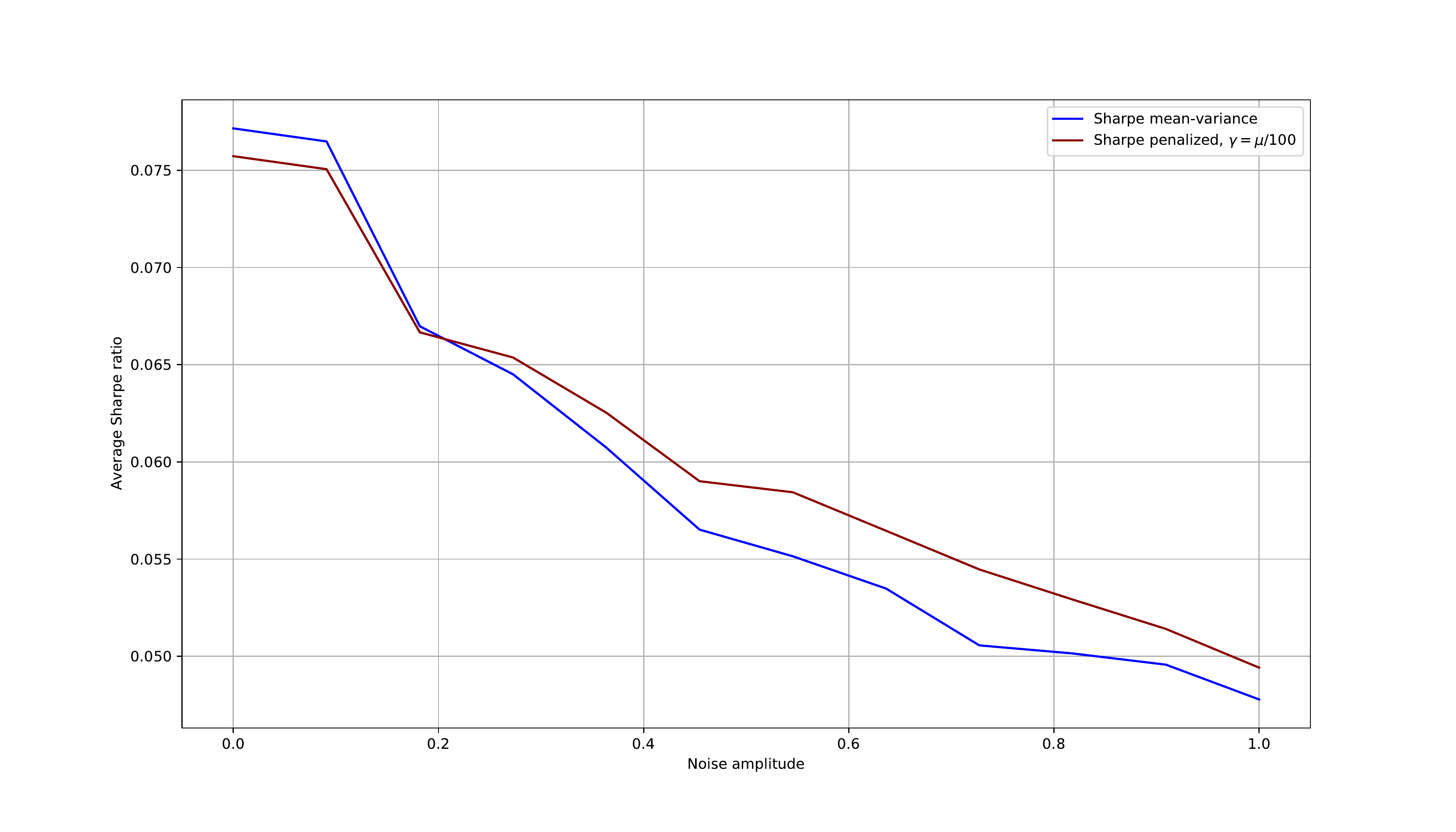}
\label{fig:simu_sharpe_equalw_ref}
\caption{The highest average Sharpe ratio attained by the equal-weight portfolio is equal to 0.047 for $\epsilon = 0$.}
\end{figure}

%\newpage

\item Minimum-variance reference portfolio
\begin{figure}[h!]
\centering
\includegraphics[width=9cm,height=7cm]{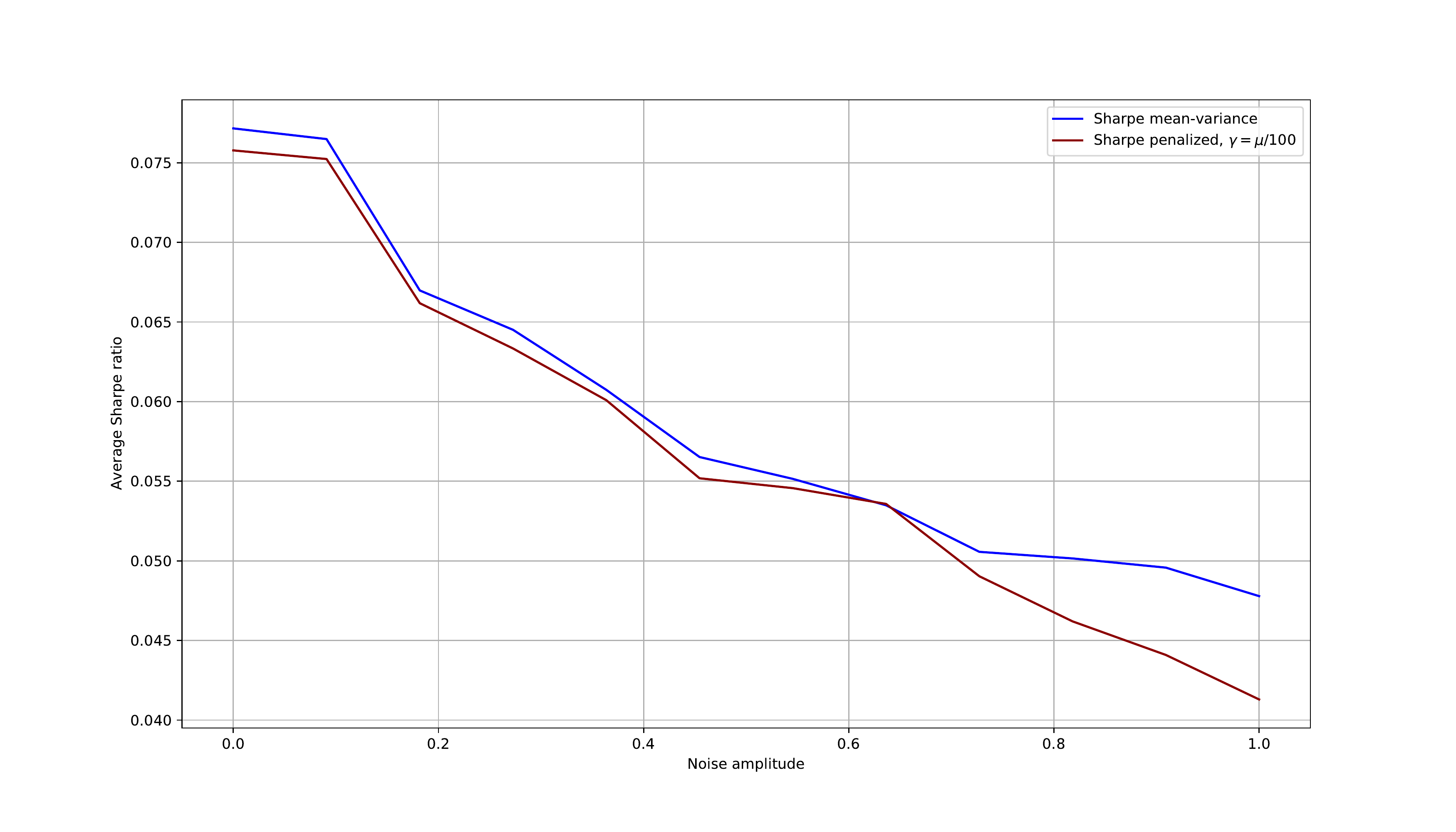}
\caption{The highest average Sharpe ratio attained by the minimum-variance portfolio is equal to 0.057 for $\epsilon = 0$.}
\label{fig:simu_sharpe_minvar_ref}
\end{figure}

 \item ERC reference portfolio
\begin{figure}[h!]
\centering
  \includegraphics[width=10cm,height=7cm]{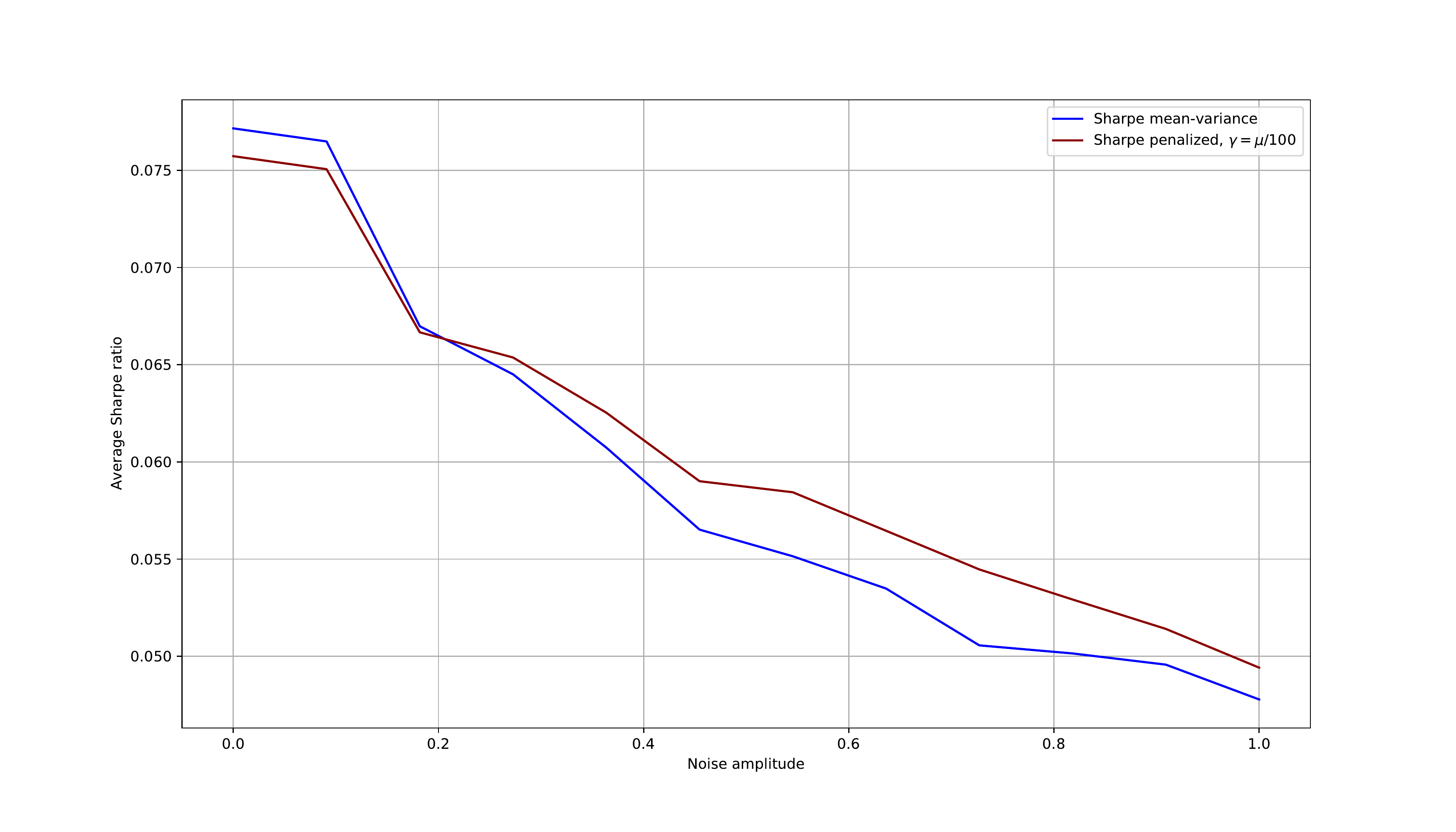}
\caption{The highest average Sharpe ratio attained by the ERC portfolio is equal to 0.051 for $\epsilon = 0$.}
\label{fig:simu_sharpe_ERC_ref}
\end{figure}

\item Control shrinking (zero reference) 
\begin{figure}[h!]
\centering
 \includegraphics[width=10cm,height=7cm]{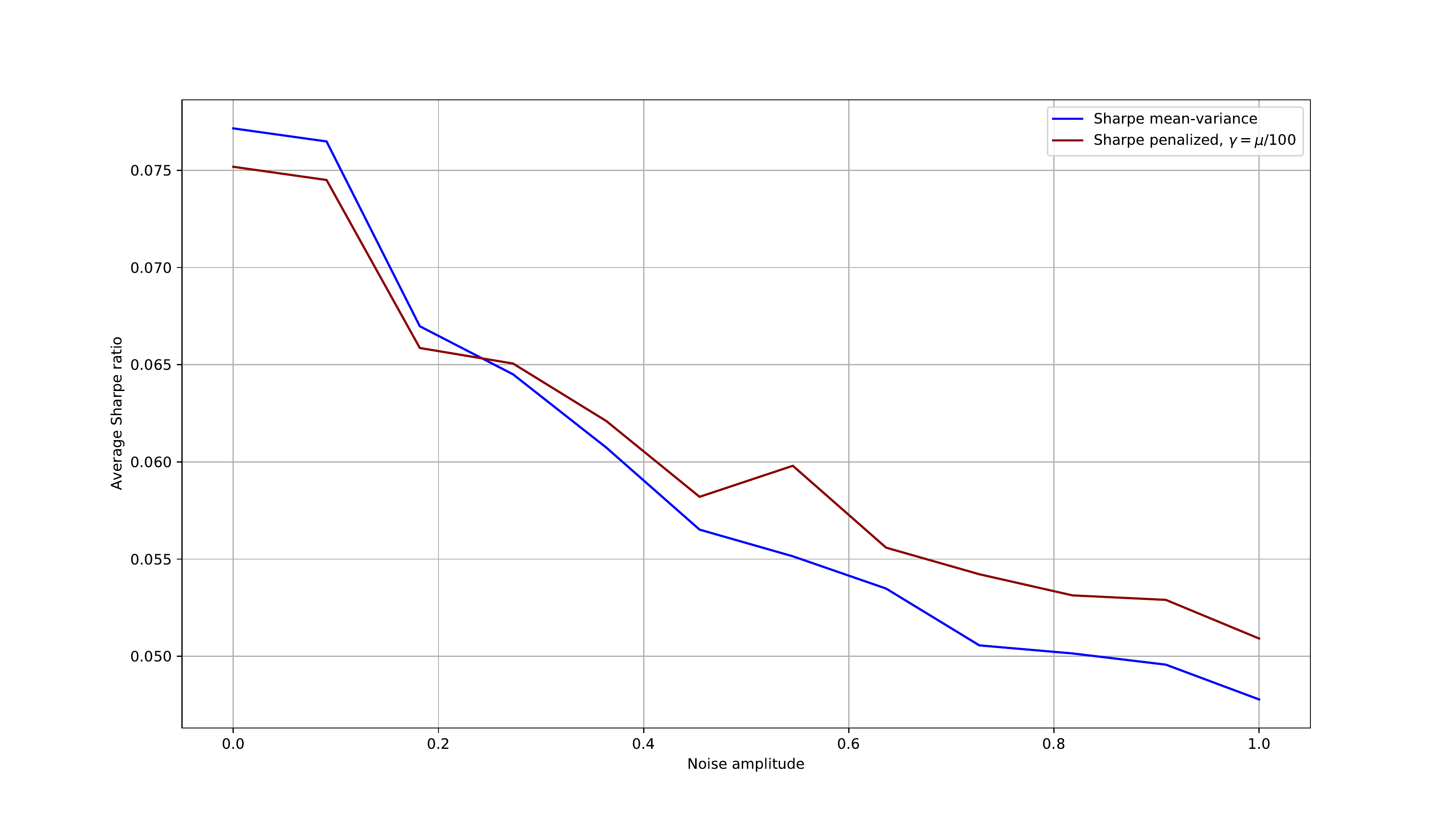}
\caption{In this case the reference weights are equal to zero, and no Sharpe ratio is computed for the reference portfolio.}  
\label{fig:simu_sharpe_zero_ref}
\end{figure}
\end{itemize}

\subsection{Performance comparison on a backtest}
We now compare the different allocations on a backtest based on adjusted close daily prices available on Quandl between 2013-09-03 and 2017-12-28 for four stocks: Apple, Microsoft, Boeing and Nike. Here we chose a value of $\mu$ which corresponds to an annual expected return of $25\%$. In our example, we express again $\gamma$ as a function of $\mu$ and we consider two different values, $\gamma = \mu$ and $\gamma = \mu/100$.

Figures \ref{fig:backtest_equalw}, \ref{fig:backtest_minvar} and \ref{fig:backtest_ERC} show the total wealth of the four different portfolios, mean-variance, reference and the penalized portfolio with the big and the small penalization as a function of time. On these graphs we observe that, at the beginning of the investment horizon, the mean-variance allocation has the largest wealth increase, hence the largest leverage. As the wealth of this portfolio attains the target wealth, expressed as $\frac{1}{2\mu}e^{b^\top \Sigma^{-1}b\ T} + x_0$ in the mean-variance control equation \eqref{eq:optimal_control_classical}, its leverage decreases and its wealth curve flattens. The same phenomenon occurs for the penalized allocation with large penalization parameter $\gamma = \mu$. In this case, the high value of the penalization parameter keeps the penalized portfolio controls close to the ones of the mean-variance portfolio.
On the contrary, the reference portfolios have constant weights and no target wealth. We can see that in each case the reference portfolio's wealth keeps increasing over the entire horizon. The wealth of the penalized portfolio with penalization parameter $\gamma=\mu/100$ follows the wealth of these reference portfolio due to the small value of the tracking error penalization.

For these three reference portfolios, we observe that the penalized portfolio with penalization parameter $\gamma=\mu$ outperforms both the mean-variance and the reference portfolios in terms of Sharpe ratio whereas the penalized portfolio with penalization parameter $\gamma=\mu/100$ outperforms the mean-variance but underperforms the reference portfolio. This can be attributed to the larger weight of the mean-variance criterion with respect to the tracking error in the optimized cost \eqref{eq:cost}  with penalization parameter $\gamma=\mu$.

Finally, Figure \ref{fig:backtest_zero} corresponds to the case of a reference portfolio with weights all equal to zero. This corresponds to a shrinking of the optimal control of the penalized portfolio. In that case, for a better visualization, we plot the total wealth of the mean-variance and penalized portfolios for penalization parameters $\gamma=\mu$ and $\gamma=\mu/100$ normalized by the standard deviation of their daily returns. On this graph, we can see that the normalized wealth of the two penalized portfolio is higher than the one of the mean-variance allocation. Similarly to the three precedent reference portfolios, the two penalized portfolios outperform the mean-variance allocation in terms of Sharpe ratio. As previously, we observe that the Sharpe ratio of the penalized portfolio with penalization parameter $\gamma=\mu$ is greater than the one with $\gamma=\mu/100$, due to the larger weight of the mean-variance criterion in the functional cost. 

\clearpage

\begin{itemize}
    \item Equal-weights reference portfolio
\begin{figure}[h!]
\centering
 \includegraphics[width=10cm,height=5.5cm]{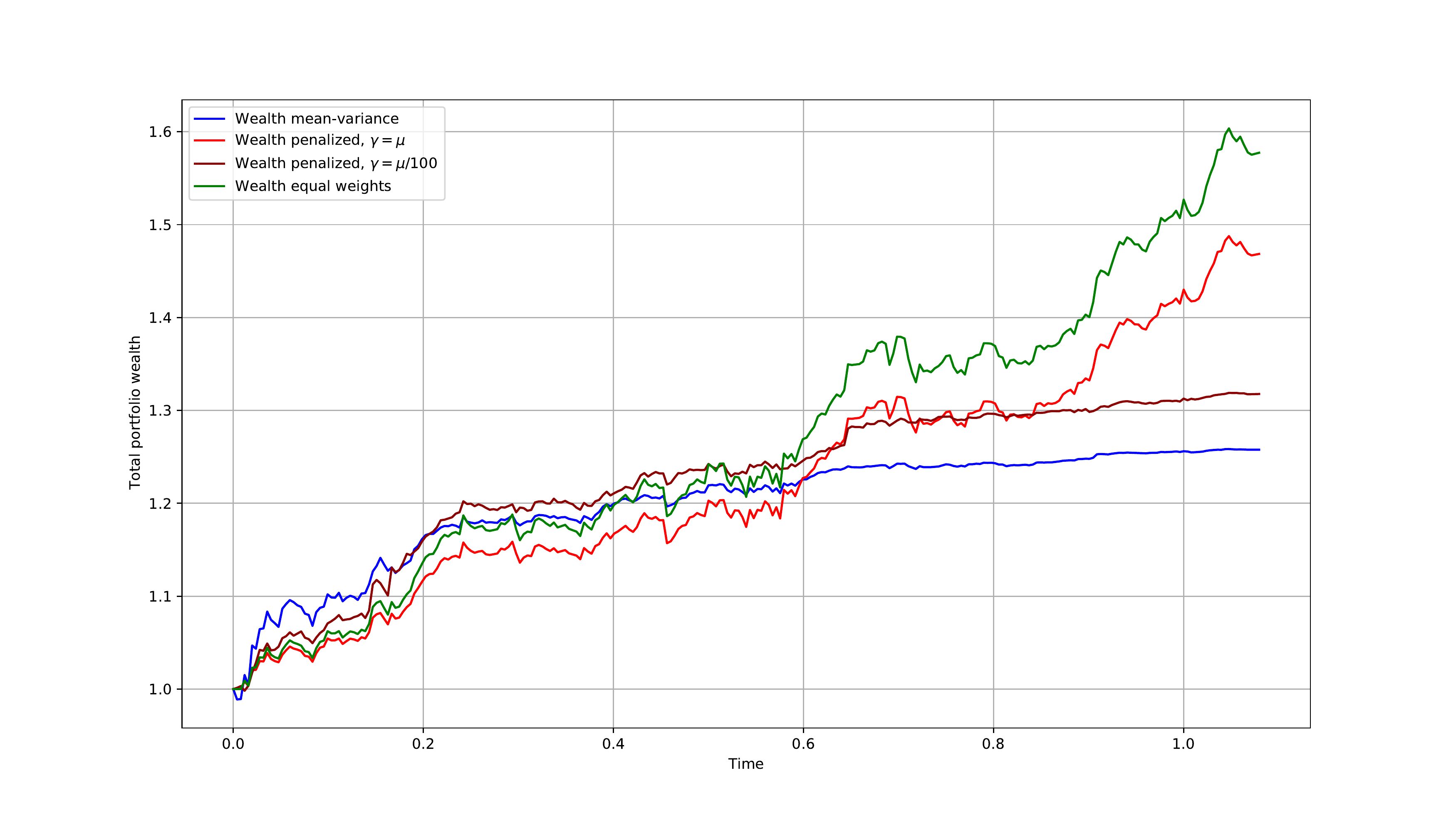}
\caption{Sharpe ratios:\\
Mean-variance : 0.183\\
Equal weights : 0.258\\
Penalized $\gamma=\mu$ : 0.260\\
Penalized $\gamma=\mu/100$ : 0.226
}
\label{fig:backtest_equalw}
\end{figure}

%\clearpage

    \item Minimum variance reference portfolio   
\begin{figure}[h!]
\centering
  \includegraphics[width=10cm,height=5.5cm]{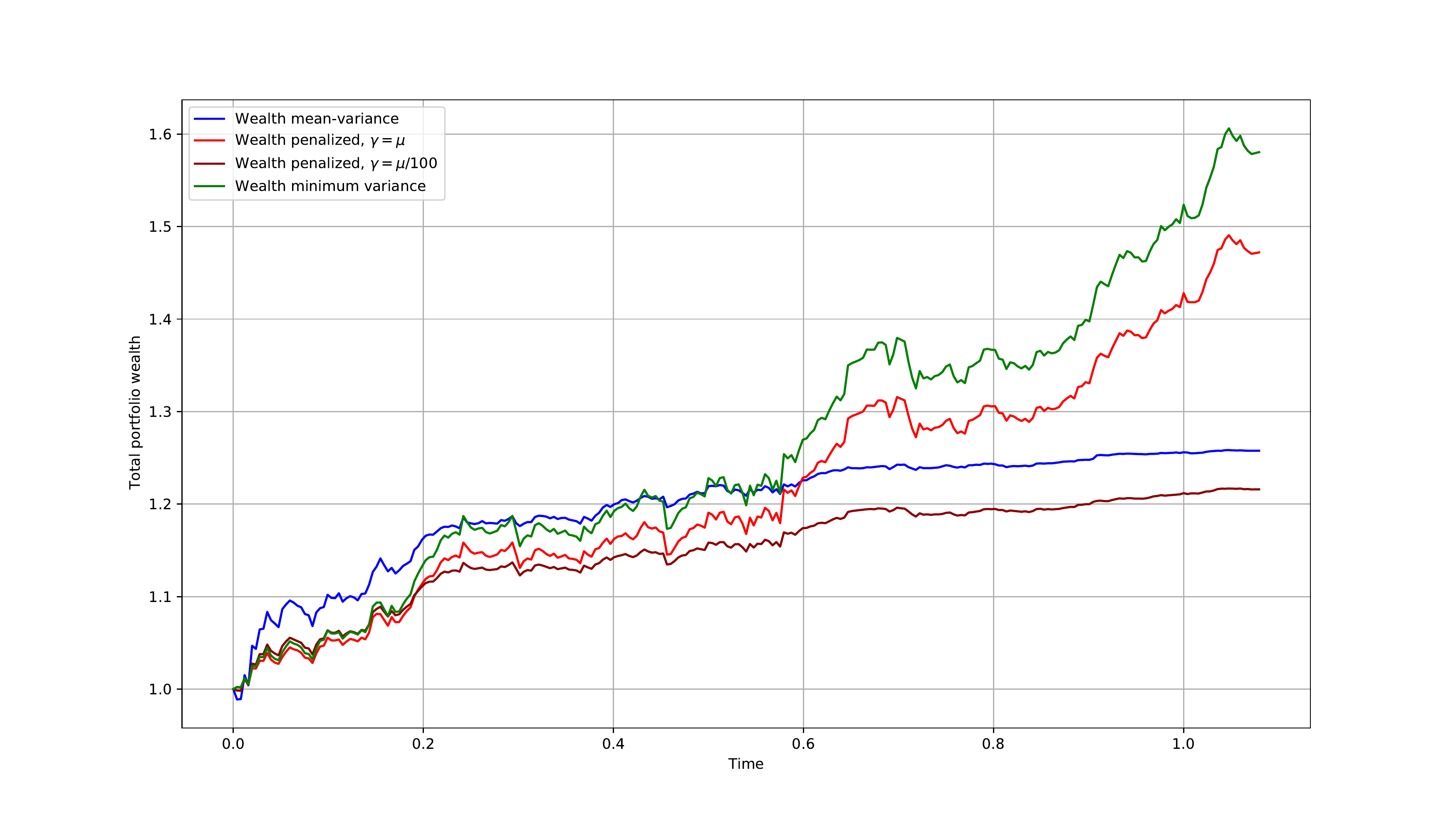}
\caption{Sharpe ratios:\\
Mean-variance : 0.183\\
Minimum variance : 0.255\\
Penalized $\gamma=\mu$ : 0.256\\
Penalized $\gamma=\mu/100$ : 0.220
}
\label{fig:backtest_minvar}
\end{figure}

\item ERC portfolio
\begin{figure}[h!]
\centering
\includegraphics[width=10cm,height=5.5cm]{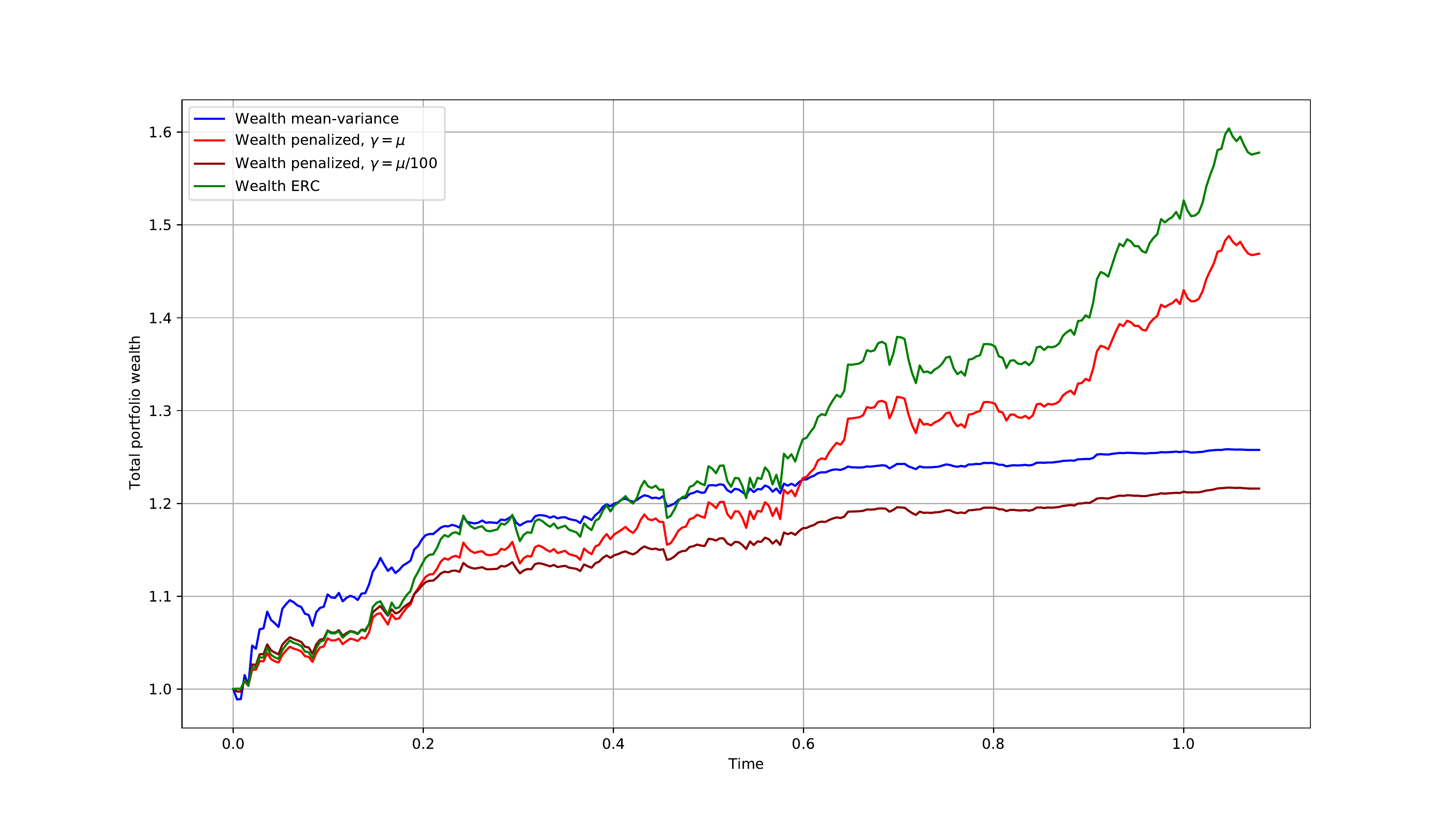}
\caption{Sharpe ratios:\\
Mean-variance : 0.183\\
ERC : 0.258\\
Penalized $\gamma=\mu$ : 0.260\\
Penalized $\gamma=\mu/100$ : 0.225
}
\label{fig:backtest_ERC}
\end{figure}

 \item Zero portfolio (shrinking)
\begin{figure}[h!]
\centering
\includegraphics[width=10cm,height=5.5cm]{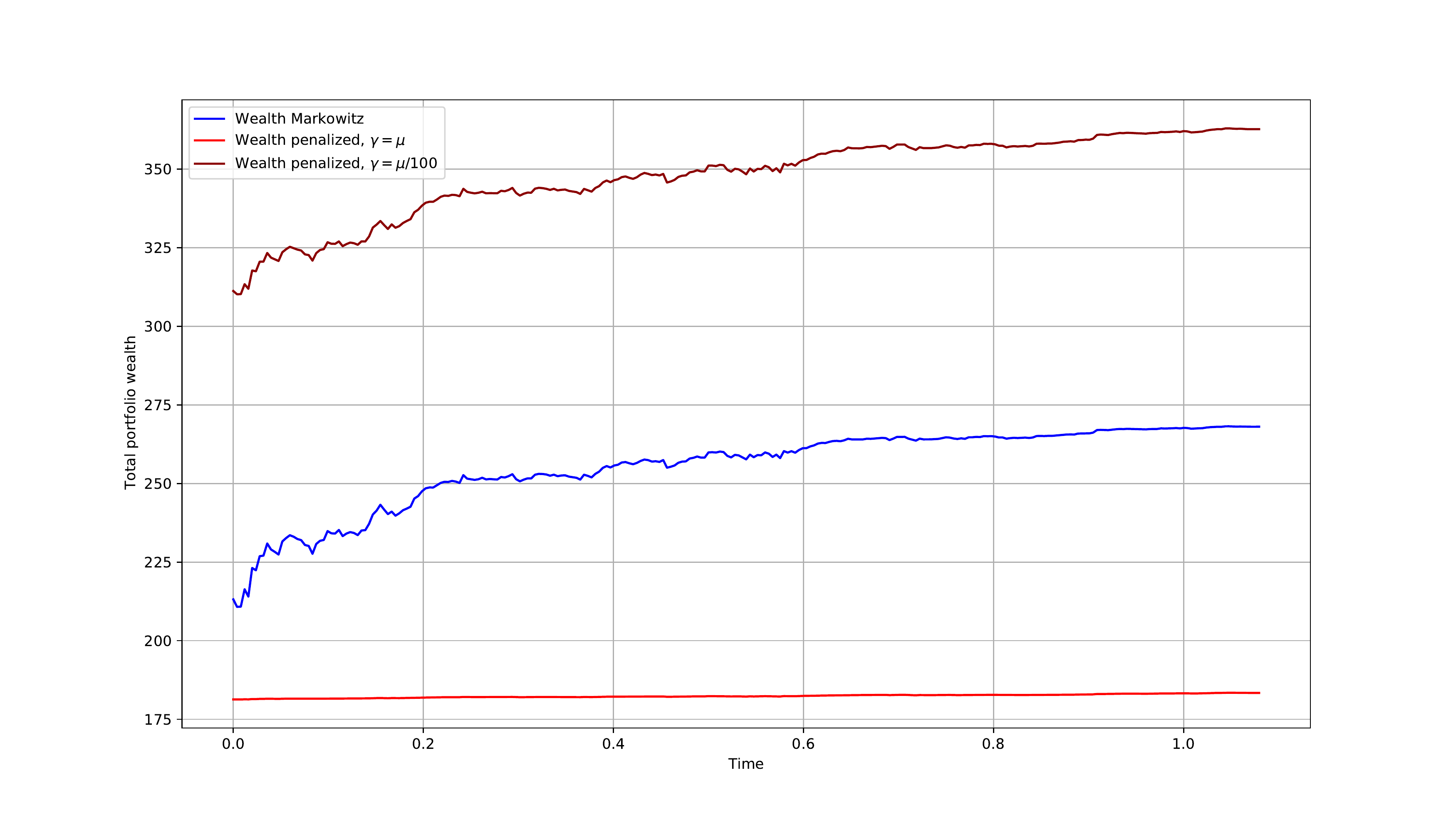}
\caption{Total wealth of the mean-variance and penalized portfolios for $\gamma=\mu$ and $\gamma=\mu/100$, normalized by the standard deviation of daily returns, as a function of time.\\
Sharpe ratios:\\
mean-variance : 0.183\\
Penalized $\gamma=\mu$ : 0.252\\
Penalized $\gamma=\mu/100$ : 0.221
}
\label{fig:backtest_zero}
\end{figure}
\end{itemize}

\section{Conclusion}
In this paper, we propose an  allocation method based on a mean-variance criterion plus a tracking error between the optimized portfolio and a reference portfolio of same wealth and fixed weights. We solve this problem as a linear-quadratic McKean-Vlasov stochastic control problem using a weak martingale approach. We then show using simulations that for a certain degree of market parameter misspecification and the right choice of reference portfolio, the mean-variance portfolio with tracking error penalization outperforms the standard mean-variance and the mean-variance allocations in terms of Sharpe ratio. Another backtest based on historical market data also shows that the mean-variance portfolio with tracking error outperforms the traditional mean-variance and the reference portfolios in terms of Sharpe ratio for the four reference portfolios considered.%\\ \\
%\textbf{Author Contributions:} . All authors have read and agreed to
%the published version of the manuscript.\\ \\
%\textbf{Funding:} This research received no external funding.\\ \\
%\textbf{Acknowledgments:} This work is issued from a CIFRE collaboration between BNP Paribas Global Markets and LPSM.\\ \\
%\textbf{Conflicts of Interest:} The authors declare no conflict of interest.

%\clearpage

\appendix

\section{Appendix}

\subsection{Proof of Theorem \ref{theorem:mv_tracking}}
\label{appendix:mv_tracking}

The proof of Theorem \ref{theorem:mv_tracking} is based on the  weak optimality principle lemma stated in \cite{basei2019weak}, and formulated in the case of the mean-variance problem \eqref{control_problem} as:

\begin{lemma}[\textbf{Weak optimality principle}]
\label{lem:weak_optimality}
Let $\left\{ V_t^\alpha , t\in [0,T], \alpha \in \mathcal{A} \right\}$ be a family of real-valued processes in the form
\begin{equation}
    V_t^\alpha = v_t(X_t^\alpha , \E[X_t^\alpha]) + \int_{0}^{t}\left(\alpha_{s}-w_rX_{s}^\alpha\right)^{\top}\Gamma\left(\alpha_{s}-w_rX_{s}^\alpha\right)ds,
\end{equation}
for some measurable functions $v_t$ on $\R \times \R$, $t\in [0,T]$, such that:
\begin{enumerate}[label=(\roman*)]
    \item $v_T(x,\bar x)$ $=$ $\mu (x - \bar x)^{2} -x$, for all $x,\bar x\in\R$, 
    \item the function $t\in [0,T]\rightarrow \E\left[ V_t^\alpha \right]$ is nondecreasing for all $\alpha\in \mathcal{A}$
    \item the map $t\in [0,T]\rightarrow \E\left[ V_t^{\alpha^{*}} \right]$ is constant for some $\alpha^* \in \mathcal{A}$.
\end{enumerate}
Then, $\alpha^*$ is an optimal portfolio strategy for the mean-variance problem with tracking error \eqref{control_problem}, and 
\begin{equation}
    V_0 = J(\alpha^*).
\end{equation}
\end{lemma}

We aim to construct a family of processes $\left\{ V_t^\alpha , t\in [0,T], \alpha \in \mathcal{A} \right\}$ as in Lemma \eqref{lem:weak_optimality}, and given the linear-quadratic structure of our optimization problem, we look for a measurable function $v_t$ in the form: 
\begin{equation}
\label{eq:candidate_field}
    v_{t}(x,\overline{x})=K_{t}(x - \overline{x})^{2}+ \Lambda_{t}\overline{x}_{}^{2}+2Y_{t}x_{}+R_{t}
\end{equation}
for some deterministic processes $\left(K_{t},\Lambda_{t},Y_{t},R_{t}\right)$ to be determined. Condition $(i)$ in Lemma \eqref{lem:weak_optimality} fixes the terminal condition 
\begin{equation}
\label{eq:terminal_condition}
    K_T = \mu, \ \Lambda_T = 0, \ Y_T = -1/2,\ R_T = 0. 
\end{equation}
For any $\alpha\in \mathcal{A}$, with associated wealth process $X:=X^\alpha$, let us compute the derivative of the deterministic function $t\rightarrow \E[V_t^\alpha] = \E \left[ v_t(X_t, \E[X_t]) + \int_{0}^{t}\left(\alpha_{s}-w_rX_{s}\right)^{\top}\Gamma\left(\alpha_{s}-w_rX_{s}\right)ds \right]$ with $v_t$ as in \eqref{eq:candidate_field}. From the dynamics of $X = X_t^\alpha$ in \eqref{eq:state_dynamic} and by applying Itô's formula, we obtain

\begin{align}
\label{eq:ansatz_value}
    \frac{d\E [V_t^\alpha]}{dt}= & \textrm{Var}({X}_{t})\left(\dot{K}_{t}+w_r^{\top}\Gamma w_r\right) + \overline{X}_{t}^{2}\left(\dot{\Lambda}_{t}+w_r^{\top}\Gamma w_r\right)+2\overline{X}_{t}\dot{Y}_{t}+\dot{R}_{t}\\
    &+ \E[G_t(\alpha)]
\end{align}
where 
\begin{equation}
    G_t(\alpha):=\alpha_{t}^{\top} S_t \alpha_{t}+2\left\{ \left(K_{t}(X_t-\overline{X}_{t})+Y_{t}+\Lambda_{t}\overline{X}_{t}\right)b^{\top}-X_{t} w_r^{\top}\Gamma\right\} \alpha_{t}.
\end{equation}
By completing the square in $\alpha$, and setting $S_t := K_t \Sigma + \Gamma$ and $\tilde{\rho}_t := b^\top S_t^{-1} b$, we rewrite $G_t(\alpha)$ as 
\begin{align}
    G_t(\alpha) =&\mathbb{E}\left[\left(\alpha_{t}-\alpha_{t}^{\Gamma}\right)^{\top}S_{t}\left(\alpha_{t}-\alpha_{t}^{\Gamma}\right)\right]\\
 & -\textrm{Var}(X_{t})\left\{ K_{t}^{2}\tilde{\rho}_t+ w_r^{\top}\Gamma S_{t}^{-1}\Gamma w_r-2K_{t}b^{\top}S_{t}^{-1}\Gamma w_r\right\} \\
 & -\overline{X}_{t}^{2}\left\{ \Lambda_{t}^{2} \tilde{\rho}_t +w_r^{\top}\Gamma S_{t}^{-1}\Gamma w_r-2\Lambda_{t}b^{\top}S_{t}^{-1}\Gamma w_r\right\} \\
 & -2\overline{X}_{t}\left\{ \Lambda_{t}Y_{t} \tilde{\rho}_t -Y_{t}b^{\top}S_{t}^{-1}\Gamma w_r\right\} -Y_{t}^{2}\rho_{t}
\end{align}
with $\alpha_t^\Gamma := S_t^{-1} \Gamma w_r X_t -S_{t}^{-1}b\left[K_{t}X_t +Y_{t}-(K_t - \Lambda_{t})\overline{X}_{t}\right]$. The expression in \eqref{eq:ansatz_value} is then rewritten as 

\begin{align}
    \frac{d\E [V_t^\alpha]}{dt}=& \mathbb{E}\left[\left(\alpha_{t}-\alpha_{t}^{\Gamma}\right)^{\top}S_{t}\left(\alpha_{t}-\alpha_{t}^{\Gamma}\right)\right]\\
 & +\textrm{Var}(X_{t})\left\{ \dot{K}_{t}-K_{t}^{2} \tilde{\rho}_t +w_r^{\top}\Gamma w_r+2K_{t}b^{\top}S_{t}^{-1}\Gamma w_r-w_r^{\top}\Gamma S_{t}^{-1}\Gamma w_r\right\} \\
 & +\overline{X}_{t}^{2}\left\{ \dot{\Lambda}_{t}-\Lambda_{t}^{2} \tilde{\rho}_t + w_r^{\top}\Gamma w_r+2\Lambda_{t}b^{\top}S_{t}^{-1}\Gamma w_r-w_r^{\top}\Gamma S_{t}^{-1}\Gamma w_r\right\} \\
 & +2\overline{X}_{t}\left(\dot{Y}_{t}+ Y_{t}b^{\top}S_{t}^{-1}\Gamma w_r-\Lambda_{t}Y_{t} \tilde{\rho}_t \right)\\
 & +\dot{R}_{t}-Y_{t}^{2}\tilde{\rho}_t.
\end{align}

Therefore, whenever

\begin{equation}
    \begin{cases}
        \dot{K}_{t}-K_{t}^{2} \tilde{\rho}_t + w_r^{\top}\Gamma w_r+2K_{t}b^{\top}S_{t}^{-1}\Gamma w_r-w_r^{\top}\Gamma S_{t}^{-1}\Gamma w_r &= 0 \\
        \dot{\Lambda}_{t}-\Lambda_{t}^{2} \tilde{\rho}_t + w_r^{\top}\Gamma w_r+2\Lambda_{t}b^{\top}S_{t}^{-1}\Gamma w_r-w_r^{\top}\Gamma S_{t}^{-1}\Gamma w_r &=0 \\
        \dot{Y}_{t}+ Y_{t}b^{\top}S_{t}^{-1}\Gamma w_r-\Lambda_{t}Y_{t}\tilde{\rho}_t &=0 \\
        \dot{R}_{t}-Y_{t}^{2}\tilde{\rho}_t &= 0
    \end{cases}
\end{equation}
holds for all $t\in[0,T]$, we have
\begin{equation}
   \frac{d\E [V_t^\alpha]}{dt}= \mathbb{E}\left[\left(\alpha_{t}-\alpha_{t}^{\Gamma}\right)^{\top}S_{t}\left(\alpha_{t}-\alpha_{t}^{\Gamma}\right)\right] 
\end{equation}
which is nonnegative for all $\alpha\in\Acal$, i.e., the process $V_t^\alpha$ satisfies the condition $(ii)$ of Lemma \eqref{lem:weak_optimality}. Moreover, we see that $V_t^\alpha = 0, \ 0\leq t \leq T$ if and only if $\alpha_t = \alpha_t^\Gamma$, $0\leq t \leq T$.\\
$X^\Gamma := X^{\alpha^\Gamma}$ is solution to a linear McKean-Vlasov dynamics and, since $K\in C([0,T], \R_+^*)$, $\Lambda\in C([0,T], \R_+)$ and $Y\in C([0,T], \R)$, $X^{\Gamma}$ satisfies the square integrability condition $\E\left[ \underset{0\leq t\leq T}{\sup} |X_t^{\Gamma}|^2 \right]<\infty$, which implies that $\alpha^\Gamma$ is $\mathbb{F}$-progressively measurable and $\int_0^T \E[|\alpha_t
^\Gamma|^2]dt<\infty$. Therefore, $\alpha^\Gamma \in \Acal$, and we conclude by the verification lemma \ref{lem:weak_optimality} that it is the unique optimal control. 
\qed

\subsection{Computation linear expansion of $\alpha^\gamma$ for $\Gamma =\gamma\mathbb{I}_d \rightarrow \mathbf{0}$}
\label{appendix:control_DL}

\begin{align}
    \alpha_{t}^{\gamma}=& \Sigma^{-1}b \left(\frac{1}{2\mu} e^{\rho T} + X_0 - X_t \right)\\
    &+\gamma \left(\Sigma^{-1}w_r+\Sigma^{-2}b \right) \frac{X_t}{K_t^0} - \gamma \|w_r\|^2 (T-t)\frac{\Sigma^{-1}}{K_t^0}b \left( X_0 + \frac{e^{\rho T}}{\mu}\left( 1-e^{-\rho t} \right) \right)\\
    &-\gamma \left( \Sigma^{-1}b X_0 C_{0,t}^1 +\Sigma^{-2}b\frac{X_0}{K_t^0} \right)\\
    &-\gamma \left( \Sigma^{-1}b \frac{H_t^1}{2} + \Sigma^{-2}b \frac{e^{\rho T}}{2K_t^0 \mu} \left( 1-e^{-\rho t} \right) \right)\\
    &-\frac{\gamma}{2}\left\{ \frac{1}{\left( K_t^0 \right)^2} \Sigma^{-1} \left( K_t^1 \mathbb{1} +\Sigma^{-1} \right)b + \Sigma^{-1}b \frac{C_{t,T}}{K_t^0} \right\} + O(\gamma^2)\\
    =& \Sigma^{-1}b \left(\frac{1}{2\mu} e^{\rho T} + X_0 - X_t \right)\\
    &+\gamma \Sigma^{-1}w_r \frac{X_t}{K_t^0}\\
    &- \gamma \Sigma^{-1}b \left\{ \frac{\|w_r\|^2}{K_t^0} (T-t)\left( X_0 + \frac{e^{\rho T}}{\mu}\left( 1-e^{-\rho t} \right) \right) + X_0 C_{0,t}^1 + \frac{H_t^1}{2} + \frac{K_t^1}{2\left( K_t^0 \right)^2} + \frac{C_{t,T}}{2K_t^0} \right\} \\
    &-\gamma \Sigma^{-2}b \left\{ \frac{e^{\rho T}}{2K_t^0 \mu} \left( 1-e^{-\rho t} \right) + \frac{1}{2\left( K_t^0 \right)^2} \right\} + O(\gamma^2). 
\end{align}

\subsection{Computation linear expansion of $Var(X_T)$ for $\Gamma =\gamma\mathbb{I}_d \rightarrow 0$}
\label{appendix:var_computation}

We recall that the linear expansion of the optimal control can be written as 
\begin{equation}
    \alpha_{t}^{\gamma}= \Sigma^{-1}b~ \alpha_t^0 +\gamma \left( \Sigma^{-1}w_r~ \alpha_t^{1,3} - \Sigma^{-2}b~ \alpha_t^{1,2} - \Sigma^{-1}b~ \alpha_t^{1,1} \right) + O(\gamma^2)
\end{equation}
where the coefficients $\alpha_t^{1,1}$, $\alpha_t^{1,2}$ and $\alpha_t^{1,3}$ are given by \eqref{eq:control_expansion_weights}. The average total wealth of the portfolio constructed by the optimal control at time $t$ is given by the ODE 
\begin{equation}
    \label{eq:expansion_average_wealth_ODE}
    d\overline{X_{t}}= \rho\zeta-\gamma\left(\rho\alpha_{t}^{1,1}+b^{\top}\Sigma^{-2}b\alpha_{t}^{1,2}\right)+\left(\gamma\frac{b^{\top}\Sigma^{-1}w_{r}}{K_{t}^{0}}-\rho\right)\overline{X}_{t} + O(\gamma^2), \quad \overline{X_{t}} = X_0,
\end{equation}
where we set $\zeta := X_0 + \frac{1}{2\mu}e^{\rho T}$. We get the solution
\begin{align}
\overline{X_{T}}= & X_{0}e^{-\rho T}+\zeta\left(1-e^{-\rho T}\right)\\
 & +\gamma\left\{ \frac{b^{\top}\Sigma^{-1}w_{r}}{\mu}\left(T\zeta- \frac{1}{2\mu}e^{\rho T} \frac{1-e^{-\rho T}}{\rho}\right)-\int_{0}^{T}\left(\rho\alpha_{s}^{1,1}+b^{\top}\Sigma^{-2}b\alpha_{s}^{1,2}\right)e^{-\rho\left(T-s\right)}ds\right\} +O\left(\gamma^{2}\right)\\
= & \overline{X_{T}}^{0}+\gamma\overline{X_{T}}^{1}+O(\gamma^{2})
\end{align}
with
\begin{equation}
\begin{cases}
\overline{X_{T}}^{0}:=X_{0}e^{-\rho T}+\zeta\left(1-e^{-\rho T}\right)\\
\overline{X_{T}}^{1}:=\frac{b^{\top}\Sigma^{-1}w_{r}}{\mu}\left(T\zeta-\frac{1}{2\mu}e^{\rho T} \frac{1-e^{-\rho T}}{\rho}\right)-\int_{0}^{T}\left(\rho\alpha_{s}^{1,1}+b^{\top}\Sigma^{-2}b\alpha_{s}^{1,2}\right)e^{-\rho\left(T-s\right)}ds
\end{cases}
\end{equation}
and
\begin{equation}
\overline{X_{T}}^{2}= \left(\overline{X_{T}}^{0}\right)^{2}+2\gamma\overline{X_{T}}^{0}\overline{X_{T}}^{1}+O(\gamma^{2})
\end{equation}

The average of the square of the portfolio wealth at time $t$ is given by the ODE
\begin{align}
d\overline{X_{t}^{2}}= & \left(\zeta-\gamma\alpha_{t}^{1,1}\right)^{2}\rho-2\gamma b^{\top}\Sigma^{-2}b\alpha_{t}^{1,2}\left(\zeta-\gamma\alpha_{t}^{1,1}\right)\\
 & +2\gamma\frac{w_{r}^{\top}\Sigma^{-1}b}{K_{t}^{0}}\left(\zeta-\gamma\alpha_{t}^{1,1}\right)\overline{X}_{t}\\
 & -\rho\overline{X_{t}^{2}} + O(\gamma^2)
\end{align}
which gives the solution
\begin{align}
\overline{X_{T}^{2}}= &X_{0}^{2}e^{-\rho T}+\zeta^{2}\left(1-e^{-\rho T}\right)\\
 & -2\gamma\zeta\int_{0}^{T}\left\{ \rho\alpha_{s}^{1,1}+b^{\top}\Sigma^{-2}b\alpha_{s}^{1,2}\right\} e^{-\rho(T-s)}ds\\
 & +2\gamma\frac{w_{r}^{\top}\Sigma^{-1}b}{\mu}\zeta\int_{0}^{T}\overline{X_{s}}^{0}ds+O(\gamma^{2}).
\end{align}
We can then compute the variance of the terminal total wealth of the portfolio given by the control \eqref{eq:optimal_control_expansion}
\begin{align}
Var(X_{T})= & \overline{X_{T}^{2}}-\overline{X_{t}}^{2}\\
= & \frac{e^{-\rho T}}{1-e^{-\rho T}}\left(\overline{X_{T}}^{0}-X_{0}\right)^{2}\\
 & +\gamma\frac{b^{\top}\Sigma^{-1}w_{r}}{\mu^{2}}\left(\zeta T-\frac{1}{2\mu}e^{\rho T}\frac{1-e^{-\rho T}}{\rho}\right)\\
 & -\gamma\int_{0}^{T}\left(\frac{\rho}{\mu}\alpha_{s}^{1,1}+\frac{b^{\top}\Sigma^{-2}b}{\mu}\alpha_{s}^{1,2}\right)e^{-\rho\left(T-s\right)}ds+O(\gamma^{2}). 
\end{align}
\qed

%\subsection{Proof that $\frac{K_t}{\gamma},\frac{\Lambda_t}{\gamma}\underset{\gamma\rightarrow\infty}{\longrightarrow 0}$}
\subsection{Proof that $K_t$ and $\Lambda_t$ are bounded in $\gamma$}
\label{appendix:bounded_quotients}

To prove this, we use a theorem from \cite{gronwall1919note} (also in \cite{hairer1993solving}, Theorem 14.1, p93).
    We rewrite the differential equation of $K$ as 
    \begin{equation}
        \label{eq:K_ODE}
        \frac{dK_t}{dt}= f(t, K_t, \gamma) ,\quad K_T =\mu
    \end{equation}
    with $f(t, K_t, \gamma) := \left( K_t b-\gamma w_r \right)^\top \left( K_t \Sigma + \gamma \mathbb{1} \right)^{-1} \left( K_t b-\gamma w_r \right) - \gamma \|w_r\|^2$, where $\|\cdot \|$ denotes the euclidean norm in $\R^d$.\\
    For $t\in [0,T]$, the partial derivatives $\partial f/\partial K$ and $\partial f/\partial \gamma$ exist and are continuous in the neighbourhood of the solution $K_t$. Then the partial derivative 
    \begin{equation}
        \frac{\partial K_t}{\partial \gamma} = \psi_t
    \end{equation}
    exists, is continuous, and satisfies the differential equation
    \begin{equation}
        \psi_t^{'}=\frac{\partial f}{\partial K}(t, K_t, \gamma) \psi_t + \frac{\partial f}{\partial \gamma}(t, K_t, \gamma).
    \end{equation}
    Recalling that the derivative of the inverse of a nonsingular matrix $M$ whose elements are functions of a scalar parameter $p$ w.r.t this parameter is equal to $\frac{\partial M^{-1}}{\partial_p} = -M^{-1} \frac{\partial M}{\partial_p} M^{-1}$, we can compute the partial derivatives $\partial f/\partial K$ and $\partial f/\partial \gamma$, and we obtain the following differential equation for $\psi$
    \begin{equation}
        \begin{cases}
        \left(\psi_t\right)^{'}= \left[ -\|\sigma^\top S_t^{-1} \left( K_t b - \gamma w_r \right) \|^2 + 2 b^\top S_t^{-1} \left( K_t b - \gamma w_r \right) \right]\psi_t - \| w_r + S_t^{-1}\left( K_t b - \gamma w_r \right) \|^2, \quad t\in [0,T]\\
        \psi_T = 0.
        \end{cases}
    \end{equation}
    This ODE has an explicit solution given by
    \begin{equation}
        \psi_t = \int_t^T A_s e^{-\int_t^s B_u du} ds
    \end{equation}
    with $A_t \geq 0,\ \forall t \in [0,T]$ equal to 
    \begin{equation}
        A_t := \frac{K_t^2}{\gamma^2} \| \left( \mathbb{I}_d+\frac{K_t}{\gamma}\Sigma \right)^{-1} \left( b+\Sigma w_r \right) \|^2 \underset{\gamma \rightarrow \infty}{\longrightarrow} 0
    \end{equation}
    and 
    \begin{align}
        B_t :=& 2 \frac{K_t}{\gamma} \left( b + \Sigma w_r \right)^\top \left( \mathbb{I}_d+\frac{K_t}{\gamma}\Sigma \right)^{-1} \left( b + \Sigma w_r \right)\\
        & - \frac{K_t^2}{\gamma^2} \left( b + \Sigma w_r \right)^\top \left( \mathbb{I}_d+\frac{K_t}{\gamma}\Sigma \right)^{-1}\Sigma \left( \mathbb{I}_d+\frac{K_t}{\gamma}\Sigma \right)^{-1} \left( b + \Sigma w_r \right)\\
        & -2b^\top w_r - \|\sigma^\top w_r \|^2.
    \end{align}
    We have $B_t \underset{\gamma\rightarrow\infty}{\rightarrow}-2b^\top w_r - \|\sigma^\top w_r \|^2 $, thus $\psi_t \underset{\gamma\rightarrow\infty}{\longrightarrow}0,\ \forall t \in [0,T]$ and $K_t$ is bounded in $\gamma$ for every $t\in [0,T]$.\\
    In the same spirit, we rewrite the differential equation of $\Lambda_t$ as 
    \begin{equation}
        \frac{d\Lambda_t}{dt}=g(t, \Lambda_t, \gamma), \quad \Lambda_t = 0
    \end{equation}
    with $g(t, \Lambda_t, \gamma):=\left(\Lambda_t b -\gamma w_r \right)^\top S_t^{-1}\left(\Lambda_t b -\gamma w_r \right)-\gamma\|w_r\|^2$. The partial derivative
    \begin{equation}
        \frac{\partial\Lambda_t}{\partial\gamma}=\phi_t
    \end{equation}
    exists, is continuous and satisfies the differential equation 
    \begin{equation}
        \begin{cases}
            \phi_t^{'} = 2b^\top S_t^{-1} \left( \Lambda_t b -\gamma w_r \right) \phi_t - \left[\| w_r + S_t^{-1} \left( \Lambda_t b -\gamma w_r \right)\|^2 + \psi_t \| \sigma^\top S_t^{-1} \left( \Lambda_t b - \gamma w_r \right) \|^2 \right], \quad t\in [0,T]\\
            \phi_0 = 0.
        \end{cases}
    \end{equation}
    which gives the explicit solution 
    \begin{equation}
        \phi_t = \int_t^T C_s e^{-\int_t^s D_u du}ds
    \end{equation}
    with $C_t\geq 0,\ \forall t \in [0,T]$ equal to 
    \begin{align}
        C_t := \| \frac{1}{\gamma} \left( \mathbb{I}_d +\frac{K_t}{\gamma}\Sigma \right)^{-1}\left( \Lambda_t b +K_t \Sigma w_r \right) \|^2 + \psi_t \| \frac{1}{\gamma}\sigma^\top \left( \mathbb{I}_d +\frac{K_t}{\gamma}\Sigma \right)^{-1}\left( \Lambda_t b +K_t \Sigma w_r \right) -\sigma^\top w_r \|^2
    \end{align}
    and
    \begin{equation}
        D_t := 2\left[\frac{1}{\gamma} b^\top  \left( \mathbb{I}_d +\frac{K_t}{\gamma}\Sigma \right)^{-1}\left( \Lambda_t b +K_t \Sigma w_r \right) - b^\top w_r  \right].
    \end{equation}
    We showed that $\frac{K_t}{\gamma}, \psi_t \underset{\gamma\rightarrow\infty}{\longrightarrow}0$ for every $t\in[0,T]$. Thus $C_t\underset{\gamma\rightarrow\infty}{\longrightarrow}0$, $D_t \underset{\gamma\rightarrow\infty}{\longrightarrow} -2b^\top w_r$ and $\phi_t \underset{\gamma\rightarrow\infty}{\longrightarrow} 0$, $\forall t \in [0,T]$. $\Lambda_t$ is then bounded in $\gamma$ for every $t\in [0,T]$.

 \newpage
\bibliographystyle{plainnat}
\bibliography{bibl}

\end{document}